\documentclass[twoside,leqno,twocolumn]{article}
\pdfoutput=1
\usepackage{SIAMstyle/style}

\usepackage{SIAMproceedings_060921/ltexpprt}

\newcommand{\bbox}{\text{box}}

\usepackage{subcaption}
\usepackage{tikz}
\usepackage{pgfplots}
\pgfplotsset{compat=1.7}

\begin{document}
\title{Parallel Nearest Neighbors in Low Dimensions with Batch Updates}

\author{Guy E.\ Blelloch \thanks{Carnegie Mellon University.} \and Magdalen Dobson \thanks{Carnegie Mellon University.}}

\date{}

\maketitle

\fancyfoot[R]{\scriptsize{Copyright \textcopyright\ 2022 by SIAM\\
		Unauthorized reproduction of this article is prohibited}}


\begin{abstract} \small\baselineskip=9pt
We present a set of parallel algorithms for computing exact k-nearest neighbors in low dimensions. Many k-nearest neighbor algorithms use either a kd-tree or the Morton ordering of the point set; our algorithms combine these approaches using a data structure we call the \textit{zd-tree}. We show that this combination is both theoretically efficient under common assumptions, and fast in practice. For point sets of size $n$ with bounded expansion constant and bounded ratio, the zd-tree can be built in $O(n)$ work with $O(n^{\epsilon})$ span for constant $\epsilon<1$, and searching for the $k$-nearest neighbors of a point takes expected $O(k\log k)$ time. We benchmark our k-nearest neighbor algorithms against existing parallel k-nearest neighbor algorithms, showing that our implementations are generally faster than the state of the art as well as achieving 75x speedup on 144 hyperthreads. Furthermore, the zd-tree supports parallel batch-dynamic insertions and deletions; to our knowledge, it is the first k-nearest neighbor data structure to support such updates. On point sets with bounded expansion constant and bounded ratio, a batch-dynamic update of size $k$ requires $O(k \log n/k)$ work with $O(k^{\epsilon} + \polylog(n))$ span. 
\end{abstract}

\maketitle
\section{Introduction}\label{sec: intro}

Computing nearest neighbors is one of the most fundamental problems in computer science, with applications in diverse areas ranging from graphics~\cite{clarenz2004finite, mitra2004estimating, pauly2003shape, pajarola2005stream} to AI~\cite{javier2012fast} to as far afield as particle physics~\cite{salam2006jet}. Research on nearest neighbors can be roughly divided into two areas: one area focuses on computing approximate nearest neighbors in high dimensions, primarily with clustering as an application. The second focuses on exact (or closer to exact) nearest neighbors in lower dimensions, with tasks such as surface reconstruction~\cite{alexa2001point, fleishman2005robust} as a prominent application. This work focuses on the latter category. 

The most common method of computing nearest neighbors in low dimensions is via a kd-tree~\cite{bentley1975multidimensional}, a tree which keeps the entire bounding box of the point set at its root, and whose children represent progressively smaller enclosed bounding boxes. Kd-trees have many applications in point-based graphics, and have been the data structure of choice for many graphics practitioners~\cite{pajarola2005stream}, even though other methods have better worst-case guarantees. One of the best kd-tree implementations is Arya et al's~\cite{arya1994knn}, which has been used widely by researchers~\cite{clarenz2004finite, mitra2004estimating, pauly2003shape}. Another commonly used library of kd-trees is that of the Computational Geometry Algorithms Library (CGAL)~\cite{tangelder2020spatial}.  Recent work on kd-trees has focused on better theoretical guarantees~\cite{ram2019revisiting}, and with better performance in high dimensions~\cite{chan2019fast}.  

Another approach for computing nearest neighbors uses space-filling curves known as the Morton ordering, z-ordering, or Lebesgue ordering (henceforth Morton ordering).  Recursing by splitting the Morton ordering roughly splits space, making it possible to effectively search for nearest neighbors. Two nearest neighbor algorithms that make use of Morton ordering are Chan's minimalist nearest neighbor algorithm~\cite{chan2006minimalist}, and Connor and Kumar's k-nearest neighbor graph algorithm~\cite{connor2010knn}.  Other approaches to computing nearest neighbors include well-separated decompositions~\cite{callahan1995decomposition}, and Delaunay triangulation~\cite{birn2010simple}.

Some important considerations when choosing a k-nearest neighbors algorithm are how it performs (theoretically as well as practically), does it run efficiently in parallel (since todays machines only have multiple processors), what kind of point sets it handles, and whether it supports dynamic updates (since in many applications point sets change over time~\cite{singh2021fresh}).  Vaidya~\cite{Vaidya86} and Callahan and Kosaraju~\cite{callahan1995decomposition} give strong bounds for general point sets computing all nearest neighbors in $O(n \log n)$ time using variants of kd-trees. Chan improved this to $O(n)$ time if the the ratio of the largest distance to the smallest is polynomially bounded~\cite{chan2008well}.  However these results are limited to static point sets and have not  yet shown to be practical.  Connor and Kumar give bounds under the assumption of bounded expansion constants~\cite{connor2010knn} for a practical algorithm they implement.  There has also been significant interest in parallel algorithms for the problem.  This includes implementations based on MapReduce~\cite{agarwal2016parallel}, for GPUs~\cite{hu2015massively}, the STANN library~\cite{connor2010knn}, and an implementation in CGAL~\cite{alliez2016cgal}.  Although in principle kd-trees should be able to support dynamic updates we know of no libraries that efficiently support them, and few interesting theoretical bounds for the problem in low dimensions.  When considering parallelism and updates together one should be interested in batches of updates that can be processed in parallel.  

In this paper, we present a technique that combines the ideas of kd-trees and Morton ordering to achieve efficient algorithms for k-nearest neighbors in bounded dimension. Some guiding intuition for such a combination is that Morton-based algorithms tend to have quick preprocessing (since only a sort is required) and slower queries; on the other hand, tree-based algorithms can have slower building times but their additional structure leads to faster queries. Thus, combining these approaches may allow us to achieve the advantages of both. In particular we present a k-nearest neighbor algorithm that hybridizes the kd-tree and Morton order approaches by using a kd-tree whose splitting rule is based on the Morton ordering; we call this tree the \textbf{zd-tree}. We also present what is to our knowledge the first parallel batch-dynamic update algorithm for a k-nearest neighbor data structure. We prove the following theoretical results in the context of a point sets with bounded expansion constant and bounded ratio, two reasonable and broadly used assumptions when computing nearest neighbors~\cite{karger2002finding,kazana2013enumeration,segoufin2017constant, gago2009bounded,beygelzimer2006cover,connor2010knn,anagnostopoulos2015low,anagnostopoulos2018randomized, chan2008well, arya1994knn}.

The first result concerns the work and span required to build the zd-tree: 
\begin{oneshot}{Theorem~\ref{thm: treebuild}}
For a point set $P$ of size $n$ with bounded ratio, the zd-tree can be built using $O(n)$ work with $O(n^{\epsilon})$ span, and the resulting tree height $O(\log n)$.
\end{oneshot}

The second result bounds the work for a k-nearest neighbor query on the zd-tree.

\begin{oneshot}{Theorem~\ref{thm: querytime}}
For a zd-tree representing a point set $P$ of size $n$ with bounded expansion, finding the k-nearest neighbors of a point $p \in P$ requires expected $O(k \log k)$ work.
\end{oneshot}
These two theorems together imply a linear-work algorithm for finding the k-nearest neighbors among a set of points (i.e. the k-nearest neighbor graph). They also imply that for a point $q \not \in P$, finding the nearest neighbors requires $O(\log n + k \log k)$ work.

The third result bounds the work and span for batches of updates.
\begin{oneshot}{Theorem~\ref{thm: batchupdate}}
Let $T$ be a pruned zd-tree representing point set $P$, and let $Q$ be a point set of size $k$, such that $|P|+|Q|=n$. Then if $P \cup Q$ and $Q$ both have bounded expansion and bounded ratio in the same hypercube $X$, $Q$ can be inserted into $T$ in $O(k\log(n/k))$ work and $O(k^{\epsilon} + \polylog(n))$ span.  
\end{oneshot}

In additional to the theoretical contributions, we implement both our nearest neighbor searching algorithm and the batch-dynamic updates described above, and we measure our nearest neighbor searching algorithm against a large number of competitors. Our algorithms are optimized for parallelism: in addition to presenting a thread-safe data structure so that queries can be conducted in parallel, we use parallelism when recursively building or updating our kd-tree.
A snapshot of our practical results can be found in Figure~\ref{fig: bythreadall}, which compares the work needed to preprocess and query a point set across our implementation and competitors.

\begin{figure}[t]
\vspace{-.5em}
\begin{center}
%
\begin{tikzpicture}[scale=.8]
	\begin{axis}[
		legend style={font=\footnotesize},
		legend columns =2,
		xlabel={Number of Threads},
		ylabel={Work = Threads $\times$ Time},
		xmin=0, xmax=150,
		ymin=0, ymax=275,
		legend pos=north west,
		ymajorgrids=true,
		grid style=dashed,
		]
		
		\addplot[
		color=blue,
		mark = o,
		]
		coordinates {
			(1,9.228)
			(2,4.686*2)
			(8,1.177*8)
			(20,.64*20)
			(40,.305*40)
			(72,.194*72)
			(144,.157*144)
		};
		\addlegendentry{Leaf-based}

			\addplot[
		color=black,
		mark = star,
		]
		coordinates {
			(1,14.668)
			(2,7.324*2)
			(8,2.96*8)
			(20,.926*20)
			(40,.445*40)
			(72,.29*72)
			(144,.198*144)
		};
		\addlegendentry{Root-based}
		
				\addplot[
		color=red,
		mark = triangle,
		]
		coordinates {
			(1,97.2300)
			(2,50.5493*2)
			(8,14.0612*8)
			(20,5.5863*20)
			(40,2.9030*40)
			(72,1.8908*72)
			(144,1.2769*144)
		};
		\addlegendentry{Chan}
		
				\addplot[
		color = brown,
		mark = x,
		]
		coordinates {
			(1,61.2762)
			(2,33.9347*2)
			(8,13.0728*8)
			(20,9.8974*20)
		};
		\addlegendentry{KNNG}
		
			\addplot[
		color=green,
		mark = square,
		]
		coordinates {
			(1,48.8569)
			(2,24.7384*2)
			(8,7.1832*8)
			(20,3.1423*20)
			(40,1.7805*40)
			(72,1.0974*72)
			(144,.8080*144)
		};
		\addlegendentry{ParlayKNNG}
		
		\addplot[
		color=violet,
		mark = diamond,
		]
		coordinates {
			(1,10.6273)
		(2,5.7972*2)
		(8,2.8556*8)
		(18,2.1349*18)
		(36,2.1758*36)
	};
	\addlegendentry{CGAL}

	\end{axis}
\end{tikzpicture}

\end{center}\vspace{-1.5em}
\caption{\small A figure showing the work (threads $\times$ time) performed by various nearest neighbor algorithms as the number of threads increases. The k-nearest neighbor graph was computed on 10 million points from a random dataset within a 3D cube. Ideally the line for a particular algorithm would be both low (small total work) and straight (indicating more threads does not change the total work).}
\label{fig: bythreadall}
\vspace{-1em}
\end{figure}
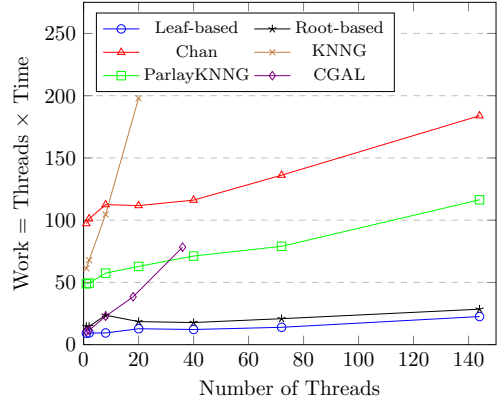

Our experimental results show the following:
\begin{enumerate}
	\item Our k-nearest neighbor algorithms achieve \textbf{high parallelism.} Using our basic algorithm to query all nearest neighbors of a 3D dataset with 10 million points achieves 75-fold speedup on a 72-core Dell R930 with 144 hyper-threads. 
	\item Our algorithms are \textbf{fast}. Our algorithm's speed is robust across all the measures which we tested---adversarial datasets, varying $k$, varying the size of the dataset, and varying the number of threads. In most cases, it beats its competitors by close to an order of magnitude. 
	\item Our batch-dynamic updates \textbf{drastically decrease the cost per insertion}. An insertion of one point into a tree
          of $5,000,000$ points takes about $10^{-5}$ seconds, while an insertion of $100,000$ points takes $10^{-7}$ seconds per pointelement.
\end{enumerate}

\subsection{Preliminaries.}\label{subsec: prelims}

For the special case where we are given a point set $P$ and wish to calculate the k-nearest neighbors of all points in $P$, we refer to the result as the \textbf{k-nearest neighbor graph} of $P$. A query of a point not in $P$ is a \textbf{dynamic query}, and a query of a point in the tree is sometimes referred to as as a \textbf{non-dynamic query}. Similarly, adding a point to or deleting a point from the tree is referred to as a \textbf{dynamic update}, which is \textbf{batch-dynamic} if the updates are processed in batches rather than one at a time.

\paragraph{Kd-trees.} Many nearest neighbor algorithms use a kd-tree as the data structure to query for nearest neighbors. Given a set of points $P$ in a $d$-dimensional bounding box, a kd-tree splits the data into two smaller bounding boxes at every level of the tree. Designing a kd-tree requires making a choice of \textbf{splitting rule}---that is, how the bounding box will be divided. One common splitting rule is to divide the bounding box of points along its largest dimension; another is to divide the space such that equal numbers of points are on each side. A variant of the kd-tree is the \textit{\{quad,oct\}-tree}, where each internal node of the tree has $2^d$ equal sized children in $d$-dimensional space.  

\paragraph{Morton ordering.} A common tool used in designing k-nearest neighbor algorithms is the Morton ordering. For a set of points whose coordinates are $d$-vectors of integers $(x_1, x_2, \ldots, x_d)$, the Morton ordering is calculated by taking each integer coordinate in binary form, interleaving the coordinates to create one integer per point, then sorting using the interleave integers. Nearest neighbor algorithms take advantage of the following property of the Morton ordering: given any two points $p$ and $q$ and the rectangle defined by those points as the corners, then all points in the rectangle must fall between $p$ and $q$ in the
Morton ordering. This allows pruning regions of the ordering.

\paragraph{Bounded expansion constant.}   
Several previous works for nearest neighbors in metric spaces have assumed a bounded expansion constant~\cite{karger2002finding, kazana2013enumeration, segoufin2017constant, gago2009bounded,beygelzimer2006cover,connor2010knn,anagnostopoulos2015low, anagnostopoulos2018randomized}, which roughly requires that the density of points in the metric space does not change rapidly.   
In the context of Euclidean space we use the following definition. Given a point $p_i$ and a positive real $r$, let $\bbox(p_i, r)$ denote the box centered at $p_i$ with radius $r$ (that is, half the side length).

\begin{definition}\label{def: expansion}
Given a point set $P$ contained in a bounded Euclidean space $X$, $P$ has \textbf{expansion constant} $\gamma$ if for all $x \in X$ and all positive real $r$, if $|\bbox(x, r)| = k$ for any $k > 1$ then
\begin{align*}
|\bbox(x, 2r)| \leq \gamma k.
\end{align*}
The expansion constant is referred to as \textbf{bounded} if $\gamma = O(1)$.
\end{definition}

\paragraph{Bounded Ratio.} Another property that will be needed to prove some of our theorems is that of the point set having bounded ratio. This is a commonly used property in problems such as nearest neighbors and closest pair~\cite{arya1994knn, chan2008well}.

\begin{definition}\label{def: ratio}
Given a point set $P$ of size $n$, let $d_{\max}$ denote the maximum distance between any two points in the set, and let $d_{\min}$ denote the minimum distance between any two points in the set. Then $P$ has \textbf{bounded ratio} if 
\begin{align*}
	\frac{d_{\max}}{d_{\min}} = \text{poly}(n).
\end{align*}
\end{definition}


\paragraph{Model of computation.} Our results for the parallel algorithms are given for the binary-fork-join model~\cite{BlellochF0020}. In this model, a process can fork two child processes, which work in parallel and when both complete, the parent process continues.  Costs are measured in terms of the work (total number of instructions across all processes) and span or depth (longest dependence path among processes).  Any algorithm in the binary forking model with $W$ work and $S$ span can be implemented on a CRCW PRAM with $P$ processors in $O(W/P + S)$ time with high probability~\cite{ABP01,blumofe1999scheduling}, so the results here are also valid on the PRAM, maintaining work efficiency.

\subsection{Related Work.}\label{subsec: relatedwork}

Arya and Mount's k-nearest neighbor implementation~\cite{arya1994knn} is commonly referred to as the state-of-the-art sequential k-nearest neighbor algorithm for low dimensions. Their implementation uses a kd-tree known as a \textit{balanced box decomposition (BBD) tree}, whose splitting rule attempts to get the best of both commonly used splitting rules---that is, splitting a bounding box into approximately equal areas which also have approximately equal numbers of points. They show that using a BBD tree, an approximate k-nearest neighbor query takes $O(\frac{k}{\epsilon} \log n)$ work. The BBD tree theoretically supports insertions, but their given implementation does not. In~\cite{connor2010knn}, Connor and Kumar parallelize Arya and Mount's all-nearest neighbors implementation and show that theirs produces faster results, so we compare against Connor and Kumar's instead of theirs. 

One nearest neighbor algorithm which uses the Morton ordering instead of a kd-tree is Chan's ``minimalist" nearest neighbors algorithm~\cite{chan2006minimalist}, which has a theoretical guarantee of $O(n \log n)$ expected preprocessing time and $O(\frac{1}{\epsilon} \log n)$ expected time per query for approximate nearest neighbors. The algorithm is notable for both its simple proof and strikingly minimalist implementation, whose sequential version requires fewer than 100 lines of code in C++.   Chan's algorithm first randomly shifts the coordinate of each point, then sorts the points using the Morton ordering. The algorithm then uses an implicit tree, recursively dividing the sorted points and visiting every implicit vertex which is within some radius of the query point. An adversarial case for this algorithm is when some query point $q$ is in the right half of the sorted data and its nearest neighbor $p$ is in the left half, causing the algorithm to search a large number of vertices. The random shift helps avoid this case in expectation.

The k-nearest neighbor implementation that most closely matched ours---in that it is tailored for parallelism and for exact nearest neighbor searching in low dimensions---is that of Connor and Kumar in~\cite{connor2010knn}, where STANN stands for Simple Threaded Approximate Nearest Neighbor. Their algorithm makes several improvements on Chan's algorithm, especially for the case of computing the k-nearest neighbor graph. Their main improvement is to search from the leaf of the implicit tree rather than the root, which allows for the possibility of searching only $O(k)$ implicit nodes instead of at least $O(\log n)$ (and this would be a best case scenario where the tree is perfectly balanced). Indeed, they find that if the input point set has bounded expansion constant, their data structure uses $O(n \log n)$ work and their nearest neighbor queries use expected $O(k \log k)$ work.    Their algorithm only works for static point sets and as our experiments show is not as fast as ours.



Another well-known tree used for computing nearest neighbors is Callahan and Kosaraju's well-separated pair decomposition~\cite{callahan1995decomposition}.  For $n$ points, they can build their tree (similar to a kd-tree) in $O(n \log n)$ work and polylogarithmic span (in parallel).   Based on the tree, they can build the decomposition and find nearest neighbors in $O(n)$ work and polylogarithmic span. The approach is only described for the static case.
Another approach is to use the Delaunay triangulation of the set of points~\cite{birn2010simple}.  Although this seems to work reasonably well in two dimensions, in three and higher dimensions it can be very expensive. Beyond bounded expansion constant, another common geometric assumption used for finding nearest neighbors or the closest pair is a bound on the ratio of the furthest pair and the closest pair in a dataset~\cite{arya1994knn, clarkson1994algorithm, clarkson1999nearest, erickson2003nice,chan2008well}.


\section{Algorithm Design and Bounds}\label{sec: algo}

Here we describe our algorithms for constructing a data structure for k-nearest neighbors, querying the structure, and batch updating it.

\textbf{Data structure.} The data structure we use for nearest neighbor searching is a kd-tree whose splitting rule uses the Morton ordering; this is what we refer to as the zd-tree. Since the Morton ordering is just the interleaving of the bits of each coordinate, the tree is built by letting the root represent the entire bounding box, and splitting the points into child nodes at level $i$ based on whether the bit at place $i$ is 0 or 1.  In three dimensions, our tree is almost equivalent to an oct-tree in which every three levels of our tree corresponds to one level of the oct-tree; however, the leaves can be at different levels.
Each internal node of the tree stores the two opposing corners defining its bounding box, its two children, and its parent.  Each leaf node stores its two opposing corners, its parent, and the set of points it contains.  We bound the number of points in a leaf by a constant, and a leaf can be empty. Note that every point covered by the root bounding box is included in exactly one leaf node.

\textbf{Construction.}  Before the zd-tree can be built, we preprocess the input. Firstly, motivated by Chan~\cite{chan2006minimalist}, and necessary for our bounds (the proof of Theorem~\ref{thm: querytime}), we select a random shift for each coordinate, and shift all the coordinates by this amount.   This shift is kept throughout.  We then sort the points by the Morton order.  This can use Chan's comparison function, which leads to an $O(n\log n)$ work sort, but as we describe in Section~\ref{subsec: theory} can be reduced to a linear time radix sort with span $O(n^\epsilon)$~\cite{jaja1992parallel} when assuming a bounded expansion constant.  In this case the number of bits needed for the Morton order can bounded by $O(\log n)$.      

After shifting and sorting we apply a divide-and-conquer
algorithm (Algorithm~\ref{algo: treebuild}) to build the zd-tree.    The algorithm recurses at each level of the tree on the two sides of the cut for the given bit of the Morton ordering. Importantly, finding the cut in the routine \texttt{splitUsingBit} only requires a binary search since the points are sorted by Morton order. This implies that when the tree is sufficiently shallow (guaranteed by bounded ratio) the work to build the tree is only linear, and the parallel depth is low.   Even if completely imbalanced the work would be $O(n \log n)$.

\paragraph{Downward search algorithm.} Our downward search algorithm is detailed in Algorithm~\ref{algo: naivesearch}.  The algorithm maintains a current set of $k$ nearest neighbors, which starts empty and is improved over time by inserting closer points. In our pseudocode, we use $N$ to represent the nearest neighbor candidate set. The downward search works as follows: let $r$ be the distance from $p$ to the furthest element in $N$ if $N$ contains at least $k$ elements, or infinity otherwise. Now search vertex $v$ only if the bounding box for $v$ intersects a ball of radius $r$ around $p$.  This is determined by the \texttt{withinBox} function. If the node is a leaf, iterate through the points it contains and update the set of nearest neighbors if necessary. If it is not a leaf, recurse on its children, searching first the child whose center is closest to the query point $p$. 

Our \textbf{root-based} algorithm simply starts at the root of the zd-tree with an empty $N$ and applies \texttt{searchDown}, but we also use \texttt{searchDown} in our upward algorithm.

\paragraph{Upward search algorithm.}  
Our upward search algorithm is detailed in Algorithm~\ref{algo:searchup}. It always starts at the leaf in the tree containing the point $p$ and works its way up the tree.  The idea is that in general, only a small part of the tree needs to be examined.  It uses the downward search as a subroutine.  As in the downward search, it maintains a priority queue $N$, initially empty, of the current estimate of the nearest $k$ neighbors, which is improved over time helping to prune further search.  The algorithm starts at the leaf by adding any points in the leaf to $N$.  Then, as with the downard version, let $r$ be the distance between $p$ and $k$-th nearest neighbor in $N$, or infinity if there are not $k$ neighbors in $N$ yet.  Now search the parent of the current node if and only if the ball of radius $r$ around $p$ extends outside the bounding box of the current node.  Otherwise we know there are no points not included in the current node that could be closer than those in $N$.   This can use the same \texttt{withinBox} as used in the downward algorithm, but with a negative $r$.  When searching the parent we search the parent's other child using the downward algorithm.

Finding the leaf in which a point $p$ belongs, which is needed, depends on whether we are generating a $k$ nearest neighbor graph or using the the structure for dynamic searches for points not in the set. In the first case we know the leaf since each point is in a leaf.  Therefore to generate a k-nearest neighbor graph we need just build the tree and then run \texttt{searchUp} on each point in each leaf. We refer to this as the \textbf{leaf-based} version of our algorithm. In the second case we have to search down the tree from the root to find the location of the leaf.   This can use the bits of the Morton ordering to decide left or right.  We refer to this as the \textbf{bit-based} version, and the downward search from the root as the \textbf{root-based} version.

\textbf{Batch-dynamic updates.} The tree data structure naturally lends itself to the possibility of dynamic insertions and deletions. Insertion of a new point $q$ into a zd-tree $T$ is conceptually simple: locate the leaf of $T$ which $q$ should be inserted into; then either add $q$ to the sequence of points contained in the leaf, or if $q$ would cause the number of points in the leaf to exceed some cutoff, split the leaf into two children. This concept can be refined to a parallel batch-dynamic algorithm, which takes a set of points and recurses in parallel down the right and left children of the root. One small subtlety is that to avoid cases where an insertion might require a rebuild of the entire tree, we require a bounding box that all data will be contained in to be specified before building the initial tree. 

As with building the tree from scratch, a batch-dynamic insert starts by using the random shift to offset the points and
sorting the points to be added based on their Morton ordering. We then apply the recursive algorithm shown in Algorithm~\ref{algo: batchdynrec}.  Deletions use an almost identical algorithm.

\begin{algorithm2e}[t!]
	\caption{buildTree($P, b)$}
	\label{algo: treebuild}\small
	\SetKwBlock{ParDo}{do in parallel}{end}
	\SetAlgoLined
	\KwIn{A set of randomly shifted points sorted according to
          their Morton ordering and an integer $b$ representing the
          bit we are working on, starting with the highest bit.}
	\KwOut{The leaf or internal node that contains $P$'s bounding box}
	\uIf{$b==0$ or size($P$) $ < $ sizeCutoff}{
		\KwRet{createLeaf($P$)}
	} \Else{
		$i = \text{splitUsingBit}(P, b)$ \;
		\ParDo{ 
		$L$ = buildTree($P[1:i], b-1$) \;
		$R$ = buildTree($P[i:n], b-1$) \;
		}
		\KwRet{createInternalNode($L$, $R$)} \;
	}	
	\vspace{0.5em}
\end{algorithm2e}
\begin{algorithm2e}[t!]
	\caption{searchDown($T, p, N)$\protect\\
needed subroutines:\\
\textbf{distance}$(p,N,k) $ returns
infinity if $N$ has fewer than $k$ points and otherwise returns the distance from $p$ to the furthest point in $N$; $k=1$ if not specified.\\
\textbf{insert}$(N, p, k)$ adds $p$ to
the set $N$ keeping only the $k$ closest points to $p$.\\
\textbf{withinBox}$(T, p, r)$ returns true if $p$ is within a 
distance $r$ of the bounding box for $T$.
}
	\label{algo: naivesearch}\small
	  	\SetAlgoLined
	\KwIn{A pointer to a tree node $T$, the query point $p$, and a 
        current set of up to $k$ nearest neighbors $N$.}
	\KwOut{The $k$-nearest neighbors of $p$.}
	$r \leftarrow \text{distance}(p, N,k)$ \;
	\uIf{withinBox($T, p, r$)}{
		\uIf{$T = $ Leaf}{
			$Q \leftarrow$ set of points contained in $T$ \;
			\For{$q \in Q$}{
				\uIf{$q \not= p$}{
					\uIf{$\text{distance}(q,p) <\text{distance}(p, N, k)$}{
						insert$(N, p, k)$\;
					}
				}
			}
		}
		\Else{
			$R \leftarrow $ $T$.Right() \;
			$L \leftarrow $ $T$.Left() \;
			$\ell \leftarrow$ distance($p, L$.center())\;
			$r \leftarrow$ distance($p,R$.center()) \;
			\uIf{$\ell < r$}{
				N ' = searchDown($L, p, N$)\;
				\Return{searchDown($R, p, N'$)\;}
			} 
			\Else{
				N' = searchDown($R, p, N$)\;
				\Return{searchDown($L, p, N'$)\;}
			}
		}
	}
	
	\vspace{0.5em}
\end{algorithm2e}
\begin{algorithm2e}[t!]
	\caption{searchUp($C, p)$\protect\\
\textbf{withinBox}$(T, p, r)$ with negative $r$ returns true if $p$ is
in within the bounding box of $T$ and at least $r$ from the boundary.}
	\label{algo:searchup}\small
	\SetAlgoLined
	\KwIn{A leaf $C$ of the kd-tree and a point $p$
          within the bounding box of $C$.}
         $N = \emptyset$\\
	$Q \leftarrow$ set of points contained in $C$ \;
	\For{$q \in Q$}{
		\uIf{$q \not= p$}{
			\uIf{$d(q,p) < \text{distance}(p, N, k)$}{
				insert$(N, p, k)$\;
			}
		}
	}
	$r \leftarrow \text{distance}(p, N, k)$\;
	$P \leftarrow C$.Parent() \;
	\While{not withinBox$(C, p, -r)$ and $P \not= \top$}{
		\uIf{$P.\text{Left()} = C$}{
			$N$ = searchDown($P.\text{Right()}, p, N$)\;
		} 
		\Else{
			$N$ = searchDown($P.\text{Left()}, p, N$)\;
		}
		$C= P$\;
                $r \leftarrow \text{distance}(p, N, k)$\;
                $P \leftarrow C$.Parent() \;
	}
        \Return{$N$}
	\vspace{0.5em}
\end{algorithm2e}
\begin{algorithm2e}[t!]
	\caption{batchInsert($T, P)$}
	\label{algo: batchdynrec}\small
	\SetKwBlock{ParDo}{do in parallel}{end}
	\SetAlgoLined
	\KwIn{A pointer to a node $T$ of the kd-tree and a set of
          points $P$ contained in its bounding box and sorted
          according to their Morton order}
	\uIf{$T = $ Leaf}{
		\uIf{size($P$) + size($T$) < leafCutoff}{
			Insert $p \in P$ into $T$ \;
		} \Else{
			Split $T$ into multiple leaf nodes \;
		}
	} \Else{
		$b = T\rightarrow$bit \;
		$i = \text{splitUsingBit}(P, b)$ \;
		\ParDo{
			batchInsert($T \rightarrow \text{Right}, P[1:i]$) \;
			batchInsert($T \rightarrow \text{Left}, P[i:n]$) \;
		}
	}	
	\vspace{0.5em}
\end{algorithm2e}

%
%

\subsection{Theoretical Results.}\label{subsec: theory}
In this section, we give theoretical results on the performance of our algorithms when assuming bounded expansion constant. The results in the rest of the section assume that the point set $P$ has bounded expansion constant $\gamma \geq 2$ as well as bounded ratio, and they assume that the dimension $d = O(1)$. We also require that every coordinate of every point is unique. This is a fair assumption to make in the context of nearest neighbors, since every point set that does not have this property can be transformed into one that does and where each point retains the same nearest neighbors. The presented proofs in this section assume that $X$ is a bounding cube, but the full version of the paper contains the proofs of the same results where $X$ is any convex region. Wherever not otherwise specified, we use $B$ to denote the bounding box of the randomly shifted point set; note that the side lengths of $B$ can be at most twice the side lengths of the smallest bounding box containing $X$. 

\begin{figure*}
\vspace{-.5em}
	\begin{subfigure}{.66\columnwidth}
	\includegraphics[width=\columnwidth]{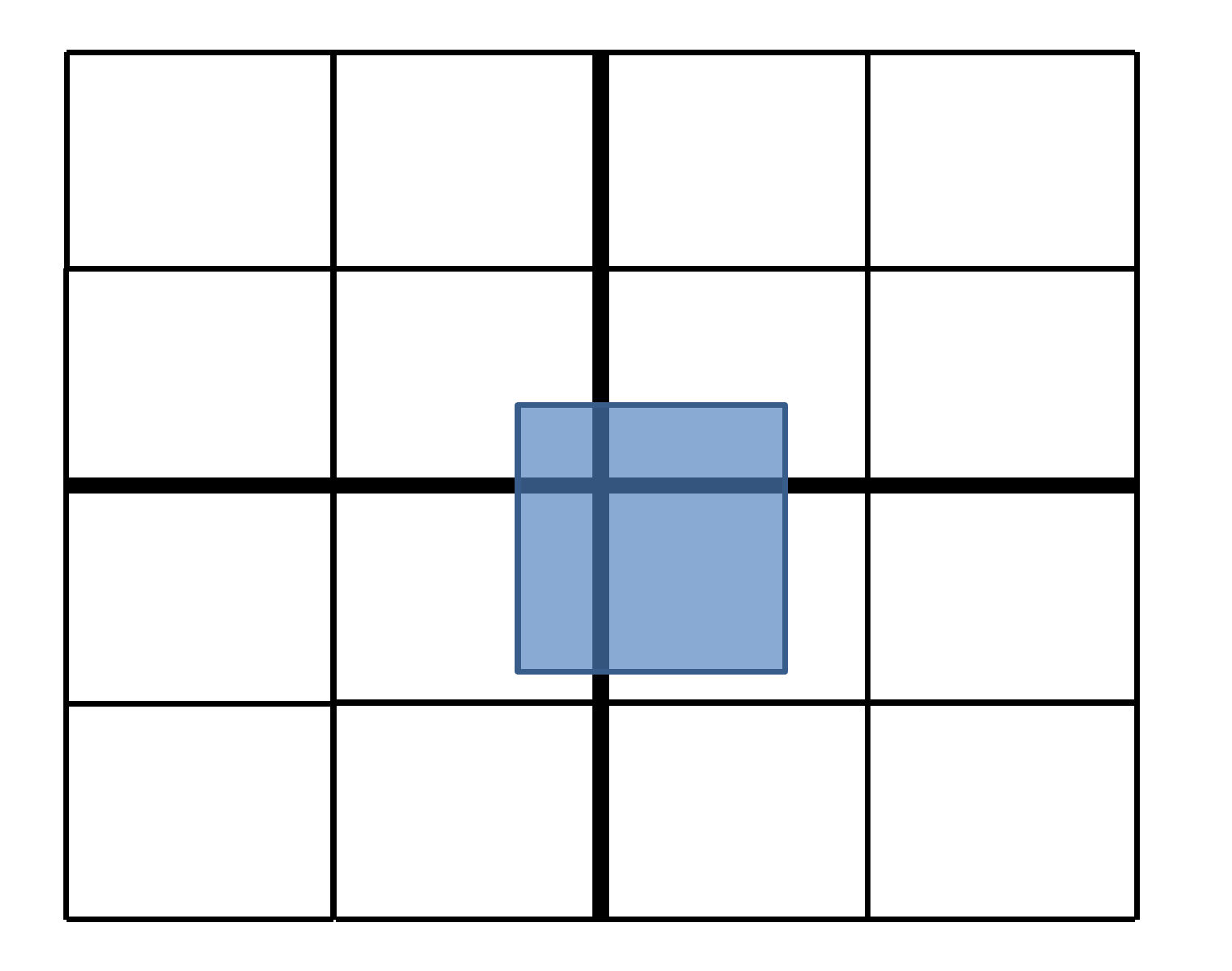}
	\caption{\small A graphic illustrating a box situated inside the largest possible quad-tree box as referenced in Lemma~\ref{lem: traversalcost}.}
	\end{subfigure}~~
	\begin{subfigure}{.66\columnwidth}
		  \includegraphics[width=\columnwidth]{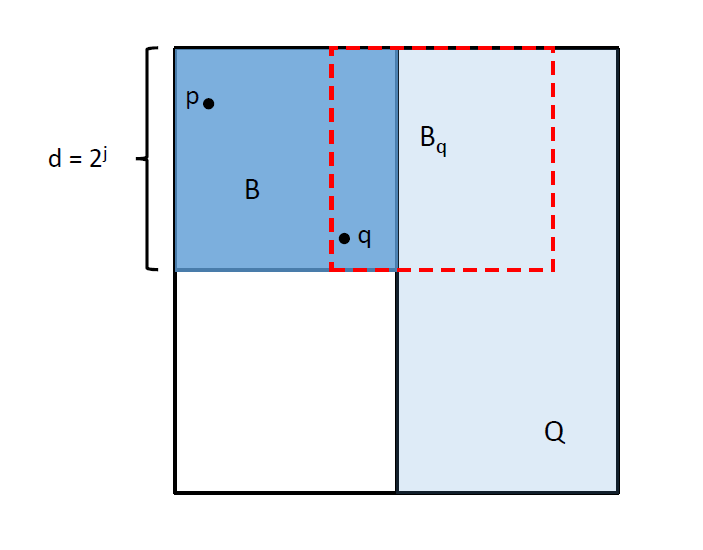} 
		\caption{\small A figure illustrating the setup in Lemma~\ref{lem: interiorcuts}. Here, the empty boxes represent empty splits, the light blue boxes represent unbalanced splits, and the dark blue box represents the box containing $p$.}
	\end{subfigure}~~
	\begin{subfigure}{.66\columnwidth}
		 \includegraphics[width=\columnwidth]{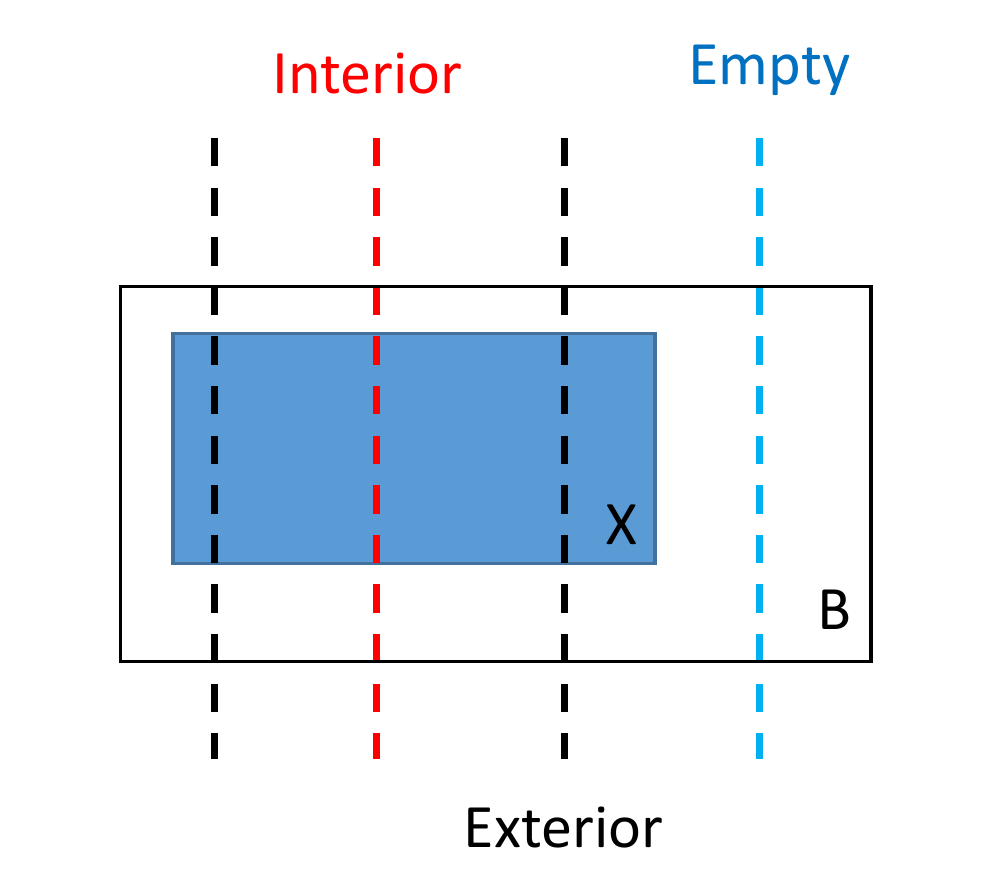} 
		\caption{\small An illustration of the three types of cuts referenced in Lemma~\ref{lem: exteriorcuts}: interior in red, exterior in black, and empty in blue.}
	\end{subfigure}~~
	
	\caption{Figures aiding with the proofs in Section~\ref{subsec: theory}.} 
  \vspace{-1em}
\label{fig: theoryfig}
\end{figure*}

Our first theorem concerns the height and build time of a zd-tree on a point set with bounded ratio.

\begin{theorem}\label{thm: treebuild}
For a point set $P$ of size $n$ with bounded ratio, the zd-tree can be built using $O(n)$ work with $O(n^{\epsilon})$ span, and the resulting tree height $O(\log n)$.
\end{theorem}

\begin{proof}
For the tree depth and work bound, we need to show that the longest path in the tree has length $O(\log n)$. The bounding cube of $X$ has side length within a constant factor of $d_{\max}$, and it must be divided until the two points whose distance between them is $d_{\min}$ are in separate cubes. Since $d_{\max}/d_{\min} = \text{poly}(n)$, $d_{\max}$ must be halved $O(\log n)$ times to reach $d_{\min}$. Thus the tree has $O(\log n)$ depth.

For the sorting claim, we wish to show that a parallel radix sort can be used. For the radix sort to require $O(n)$ work, we need to guarantee that only $O(\log n)$ bits are required to sort the dataset. This follows directly from the fact that the tree has depth $O(\log n)$.
\end{proof}

Now, we work towards the theorem on the expected time required to query the k-nearest neighbors of a point $p$. Like the result from Connor and Kumar in~\cite{connor2010knn}, the $O(k \log k)$ bound only applies when searching for the nearest neighbors of a point already in the tree; for a dynamic query, finding the leaf to start from requires $O(\log n)$ time. The following proofs assume without loss of generality that a leaf of the zd-tree contains no more than $k$ points. Furthermore, we assume that the tree is a true quad-tree, meaning that each internal node has $2^d$ children; thus we use the term ``quad-tree box" to refer to the box belonging to a tree node. This will slightly simplify the analysis, and it can only be worse than the performance of the actual algorithm.

\begin{theorem}\label{thm: querytime}
For a zd-tree representing a point set $P$ of size $n$ with bounded expansion, finding the k-nearest neighbors of a point $p \in P$ requires expected $O(k \log k)$ work.
\end{theorem}

The proof of the theorem separates the work into two parts: the work of searching through the points in a leaf of the zd-tree, and the work of traversing the zd-tree to get to those leaves. The first lemma concerns the former. 

\begin{lemma}\label{lem: numboxes}
When searching for the k-nearest neighbors of a point $p$, $O(k)$ candidate points will be considered, resulting in $O(k \log k)$ work to evaluate all the candidate points.
\end{lemma}

\begin{proof}
The multiplicative factor of $\log k$ comes from the fact that a priority queue is used to store nearest neighbors, so an $O(\log k)$ cost is incurred each time the priority queue is updated. 

Consider the leaf $L$ of the tree that would contain $p$. The initial approximation is found by recursing up to $L$'s ancestor until the ancestor has more than $k$ descendants, then adding those descendants to the priority queue. The ancestor's bounding box $B$ must contain $O(k)$ points by the fact of bounded expansion constant, since one of its children contains fewer than $k$ points.

Now, let $r$ be the side length of $B$. Our algorithm will search a leaf only if the box belonging to that leaf overlaps with $\bbox(p, r)$. Those neighbors are contained within radius $r$ of the quad-tree box containing our initial guess.

The search algorithm evaluates every point at radius $r$ from $p$ as a candidate point. If the radius $r$ of $B$ is expanded twice, all the candidate points will be contained in the resultant box; call this box of candidates $Q$. If $k' = O(k)$ points are in $B$, the expansion condition guarantees that at most $\gamma^2 k' = O(k)$ points are in $Q$. However, the search algorithm does not directly search each point in $Q$; rather, it searches every leaf whose bounding box overlaps with $Q$. This means that if a leaf $L$'s bounding box were to fall partially inside $Q$ and partially outside, all the points in $L$ would be counted. Since $B$ is a quad-tree box, this could only occur if the leaf $L$ were to have a bounding box with radius larger than $r$. Since the radius is larger than $r$ and $d=O(1)$, there can only be a constant number of such leaves. Since each leaf has at most $k$ points, the original bound still holds.
\end{proof} 

Now, we move on to considering the expected work required to traverse the zd-tree to access the leaves. The probability calculation given here can also be found in~\cite{connor2010knn}.

\begin{lemma}\label{lem: traversalcost}
When traversing the zd-tree, the expected number of tree edges traversed to find all the nearest neighbors is $O(k)$.
\end{lemma}

\begin{proof}
The worst case for the cost of the search algorithm is when the search space described in Lemma~\ref{lem: numboxes} is only contained in the bounding box containing all of $P$, as illustrated in Figure~\ref{fig: theoryfig}(a). First, we show that the length of the traversal is bounded by the longest path between two leaves in the search space; all other searches can be charged to the $O(k)$ points that are searched. We will use the box $B$ to refer to the search space. The largest cut of $B$ divides $B$ into $2^d$ quad-tree boxes. Each of these boxes must be contained within a quad-tree box $Q_i$ of at most twice the side length of $B$. Thus, traversing every leaf in $Q_i$ incurs cost at most $O(k)$; since $d$ is constant, the claim follows.

All that remains is compute the expectation of the length of the longest path in the zd-tree. Without loss of generality, assume that the search space is a box with side length $2^h$. Then the probability that the search space is contained within a box of side length $2^{h+j}$ is $\left(\frac{2^j-1}{2^j}\right)$, since the box must have its upper left corner in one of $2^j-1$ grid squares along each dimension. Thus, due to the random shift, the probability that the search space is NOT contained within a box of side length $2^{h+j}$ is $1-\left(1-\frac{1}{2^j}\right)^d$. From the perspective of traversing the zd-tree, this is the event that the path between two leaves in the search space is length $j$. Thus its expectation can be upper bounded by the following summation, which charges a cost of one for each box the search space is not contained in: 
\begin{align*}
\sum_{j=1}^\infty \left(1 - \left(1-\frac{1}{2^j} \right)^d \right) = O(1)
\end{align*}
and the result follows.
\end{proof} 

Lemmas~\ref{lem: numboxes} and~\ref{lem: traversalcost} together compose the proof of Theorem~\ref{thm: querytime}.

Now we move on to batch-dynamic updates. In a weight-balanced tree, the argument for the desired $O(k \log (n/k))$ bound would be as follows: when a batch of points is inserted into the tree, the work required to insert them into $\log k$ levels can be no more than the work that would be required to build them into their own tree. Thus insertion into the first $O(\log k)$ levels of the tree uses $O(k)$ work, and the work bound on an insertion is $O(k(\log n - \log k)) = O(\log (n/k))$. 

Hence the goal of Theorem~\ref{thm: batchupdate} is to show that in addition to the traditional notion of balance, the zd-tree also obeys some notion of weight balance---that is, that each split of a point set must produce two halves where each contains a constant fraction of the points. Unfortunately, this is not strictly true, since a set of points with bounded expansion may, for example, have only one element with a 1 at the largest bit if that element is very close to the rest of the elements. However, we will be able to show a slightly weaker notion than weight balance: that enough nodes in the tree are weight-balanced that the same work bound still applies. 

One more piece of terminology is needed before the theorem statement. When building the zd-tree by successively splitting the bounding box of the input, a split may have no points in $X$ on one side. During the tree building phase, we face a choice regarding empty cuts: when the algorithm makes an empty cut, we could either fork off two child nodes where one is a leaf containing no points, or we could simply not fork off an empty node, and divide the other node using the next bit. Since it is strictly cheaper, we choose the latter. Thus the following analysis will not deal with empty cuts; in particular, Lemmas~\ref{lem: exteriorcuts} and~\ref{lem: slicedensity} do not consider empty cuts in their analysis, since empty cuts do not affect the length of paths in the tree.

\begin{theorem}\label{thm: batchupdate}
Let $T$ be a pruned zd-tree representing point set $P$, and let $Q$ be a point set of size $k$, such that $|P|+|Q|=n$. Then if $P \cup Q$ and $Q$ both have bounded expansion and bounded ratio in the same hypercube $X$, $Q$ can be inserted into $T$ in $O(k\log(n/k))$ work and $O(k^{\epsilon} + \polylog(n))$ span. 
\end{theorem}

When $X$ in its bounding box $B$ is being recursively divided using Algorithm~\ref{algo: treebuild}, it will be useful to separate the divisions or cuts into several categories. A cut along dimension $d$ divides some sub-cube $S$ of $B$ in two with cutting plane $\ell$. The cut is either \textbf{empty}, meaning that $\ell$ does not touch any points in $X$; or it is \textbf{exterior}, meaning it touches the boundary of $X$; or it is \textbf{interior}, meaning that within $S$, $\ell$ does not touch any points on the boundary of $X$. See Figure~\ref{fig: theoryfig}(c) for an illustration.

The first step towards Theorem~\ref{thm: batchupdate} is to show that interior cuts are weight balanced.

\begin{lemma}\label{lem: interiorcuts}
One out of every $d$ interior cuts must be weight balanced; that is, it must split its bounding box into sets of size $\alpha n$ and $(\alpha-1)n$ for constant $\alpha \in (0,1)$.
\end{lemma}

\begin{proof}
Consider a point $p \in P$. As the zd-tree is built, the point set $P$ is split into smaller pieces along each dimension. Call a split \textbf{unbalanced} if it splits $P$ into pieces of size $n_1, n_2$ such that one of $n_1, n_2 < \frac{1}{2(1+\gamma)}n$. Refer to a split as ``involving" $p$ if it splits a box containing $p$. We will show that after $d-1$ unbalanced splits involving $p$, the next split must be balanced.

Assume for contradiction that there are $d$ consecutive unbalanced splits involving $p$. Let $B$ be the quad-tree box containing $p$ after those splits, and let the length of $B$'s longest side be $2^j$. By the assumption that $d$ consecutive unbalanced splits have already happened, there must be some side of $B$ where the most recent split on that side was of a region with maximum side length $2^{j+1}$; call the hyperrectangle resulting from that split $Q$. Let $q \in B$ be the unique closest point to $Q$, as shown in Figure~\ref{fig: theoryfig}(b). Then, consider any $x \in Q$ such that $B_q = \bbox(x, 2^j)$ contains $q$ and no other point in $B$, and overlaps only $B$ and $Q$. Due to its proximity to $B$, $\bbox(x, 2^{j+1})$ must completely contain $B$. By our assumption, $B_q$ contains fewer than $\frac{1}{2(\gamma+1)}n+1$ points, since it contains one point from $B$ and otherwise only points from $Q$, which has at most $\frac{1}{2(\gamma+1)n}+1$ points by our assumption. The box $B$ contains at least $\left( 1 - \frac{1}{2(1+\gamma)}\right)^d n$ points, so since $d \geq 2$, the expansion constant is violated and we reach a contradiction.
\end{proof}

While interior cuts are easily shown to be balanced, the same argument does not hold for a sequence of exterior cuts, since Lemma~\ref{lem: interiorcuts} relies on being able to choose a certain point $x$ as the center of a box, and this point might not be included in $X$ if some of the $d$ unbalanced cuts were exterior. Since an arbitrarily long sequence of nodes formed from exterior cuts might not be weight-balanced, we take a different approach: showing that even if we have to pay the maximum possible cost for each unbalanced path, the number of points on such an unbalanced path is small enough that the overall bound is unchanged. 

The first step towards this goal is to quantify, for a given point $p \in P$, how many exterior cuts involving $p$ must be made before the first interior cut involving $p$.

\begin{lemma}\label{lem: exteriorcuts}
Normalize the length of $B$ to $n$. Then for every point $p \in P$, let $f(p)$ denote the minimum distance from $p$ to the boundary of $X$, perpendicular to some side of the bounding box $B$. Then $O\left(\log \frac{n}{f(p)}\right)$ cuts will be made before the next cut containing $p$ is interior.
\end{lemma}

\begin{proof}
A cut involving $p$ is guaranteed to be interior when the quad-tree box containing $p$ has radius less than $f(p)$. Since the radius of the quad-tree box is halved every $d$ cuts and side length of $B$ is $n$, the aforementioned condition is met after $d \cdot \log(n/f(p))$ cuts.
\end{proof}

Lemma~\ref{lem: exteriorcuts} gives us a way to bound the number of exterior cuts along the path to $p$. The next step is to bound the \textit{number} of points that can be a given distance or closer to the boundary of one of the faces.

\begin{lemma}\label{lem: slicedensity}
Normalize the side of $B$ to $n$. Let $S$ be a subset of $B$ formed by cutting $B$ parallel to one of its faces at distance $r$ away from the face. Then at most a $\left(\frac{\gamma^2}{\gamma^2+1} \right)^{\log n/r}$ fraction of the total points in $B$ are contained in $S$.
\end{lemma}

\begin{proof}
One way of upper bounding the number of points in $S$ is as follows: $S$ can be formed by dividing the bounding box in half along one dimension $\log(n/r)$ times. On each division, at most how many points can be in the resultant rectangle? Consider the first cut which divides the bounding cube in half; the goal is to maximize the number of points in one half without violating the expansion constant. The ``sparse" half can be separated into $2^{d-1}$ hypercubes of radius $n/2$, each of which is expanded twice before the whole space $X$ is encompassed. Note that it is optimal for all the ``sparse" sub-cubes to contain the same amount of points since each must be able to expand a box around its center. The following equation solves for the largest fraction $f$ of the total points that can be in any one sub-cube without violating the expansion constant.
\begin{align*}
\gamma^2 \frac{1-f}{2^{d-1}-1} \geq x + (2^{d-1}-2)\frac{1-f}{2^{d-1}-1} \\
\implies f \leq \frac{\gamma^2-2^{d-1}+2}{\gamma^2+1}.
\end{align*}
The following cuts need to be upper bounded in a slightly different way, since they are cutting a hyperrectangle with $d-1$ sides of length $n$ and one side of length $s$. Thus the area can be decomposed into $2n/s$ boxes of side length $s/2$. Half the boxes will receive the maximum number of points possible; both halves will evenly distribute their points across the boxes. Each ``sparse" hypercube of side length $s/2$ must be expanded twice before it contains a cube $C$ of side length $s$ that is completely contained within the region. For that cube $C$, the same equations as before can be written to bound the number of points in the dense half of $C$ to a $\frac{\gamma^2-2^{d-1}_2}{\gamma^2+1}$ fraction of the total points in $C$. The total number of points in $C$ is upper bounded by $\left(\frac{\gamma^2-2^{d-1}_2}{\gamma^2+1} \right)^{k}$ where $k$ is the total number of cuts that have occurred, assuming that the maximum possible fraction is on one side each time. The overall bound follows.
\end{proof}

Now, we can put all these pieces together to prove the theorem.

\begin{proof}[Proof of Theorem~\ref{thm: batchupdate}]
	
The sorting cost bound follows directly from Theorem~\ref{thm: treebuild}.

For the update bound, we know that over all weight-balanced paths in the tree, the cost of insertion the $k$ points down those paths is $O(k \log (n/k))$. Thus our task is to account for the paths in the tree that are not weight-balanced. In the worst case, for every non-weight-balanced path of length $\ell$, we incur an $O(\ell)$ cost for each point that traverses it. We will show that the number of points in $Q$ that are distributed among the unbalanced parts of the tree is small enough that the overall bound is unchanged. Consider one face $S$ of $X$. For a point $p$ at a given length $r$ away from the boundary of $x$ perpendicular to $S$, the path traveled from the root to $p$ can encounter $O(\log (n/r))$ unbalanced nodes. Consider all points at a distance $r$ or closer to $S$. Lemma~\ref{lem: slicedensity} shows that there are at most $k \cdot\left(\frac{\gamma^2}{\gamma^2+1} \right)^{\log (n/r)}$ such points. The cost incurred for each depth is also $\log (n/r)$. Thus, the maximum cost for all the points at depth $f(n)$ or smaller is $k\left(\frac{\gamma^2}{\gamma^2+1} \right)^{\log (n/r)} \log (n/r)$. Renaming $\log (n/r)$ as the variable $x$ and integrating over values of $x$ from $0$ to $n/2$ gives:
\begin{align*}
	k \int_0^{n/2} \left(\frac{\gamma^2}{\gamma^2+1}\right)^x x \; dx \\
	< k \int_0^{\infty} \left(\frac{\gamma^2}{\gamma^2+1}\right)^x x \; dx \\
	= O(k).
\end{align*}
Since the number of faces is constant, the overall bound follows. This shows that even in the worst case where every exterior cut is unbalanced and the maximum number of points are distributed in nodes formed from exterior cuts, the extra work incurred does not change the overall bound.
\end{proof}

%
%
%

We conclude the section with a note on the tradeoff between work and span. The versions of the theorems given in this section have a span that is greater than polylogarithmic due to the radix sort. If the algorithms used a comparison sort, they could run in polylogarithmic span at the cost of needing $O(n\log n)$ work for the sort.

\section{Implementation Details}\label{sec: other_impls}

In this section we give more details on the practical implementation of both our algorithms and the other algorithms we use to benchmark our code.

\subsection{Our Algorithms.}\label{subsec: our_algs}

We implemented our algorithms in C++ using the parallel primitives from ParlayLib~\cite{parlay20}. Our search implementation closely matches the algorithms shown in Section~\ref{sec: algo}, so here we focus mostly on implementation details and other optimizations.

\paragraph{Numerical details.} We work with double-precision floats, which we round to 64-bit integers for building the tree and computing the Morton ordering. 

\paragraph{Miscellaneous optimizations.} We mention a few  optimizations that made significant differences in our runtime. Whenever possible, we used squared distances instead of Euclidian distances in our computations, which made our code about 10\% faster.  To store the current k-nearest neighbors when traversing the tree we use a vector for small $k$ and the C++ STL priority queue for larger $k$.  The overhead of the vector was significantly less for small $k$, but the linear instead of logarithmic cost dominates for $k > 40$ or so. A third optimization was to sort the sequence of queries using their Morton ordering so that nearby queries in this order access nearby nodes in the tree, thus reducing cache misses. The savings from reducing cache misses more than compensates for the cost of the sort, in some cases decreasing runtime by a factor of two.  This is useful even when querying in parallel since the parallel scheduler processes chunks of the iteration space on the same core.

\subsection{Other Implementations.}\label{subsec: other_impls}

For the purpose of comparison we use three existing implementations of nearest neighbor search: CGAL~\cite{alliez2016cgal}, STANN~\cite{connor2010knn}, and Chan~\cite{chan2006minimalist}. Here we describe some performance issues with their code, and some modifications we made to improve the performance of their code to ensure a fair comparison. An extended version of this discussion can be found in the full version of this paper.

\paragraph{Chan.} Chan's code was fully sequential so we needed to parallelize it. Conceptually this is relatively straightforward since the algorithm just requires using a parallel sort instead of a sequential one, and then running the queries in parallel. Chan's code only searches from the root of his implicit tree. We note that the root-based implementation of our code is significantly faster than Chan's. 

\paragraph{STANN.}  STANN includes both a k-nearest neighbor graph (KNNG) function and a k-nearest neighbor (KNN) function.  The first finds the $k$ nearest neighbors among a set of points, and the second supports a function to build a tree and a separate function to query a point for its $k$ nearest neighbors. They supply a parallel version of KNNG, that was parallelized with OpenMP, and only a sequential version of KNN. Their algorithm did not scale well beyond 16 threads, since it left some components sequential. We therefore updated their code to use the parallel primitives and built-in functions from ParlayLib~\cite{parlay20}; this drastically improved their performance.

\paragraph{CGAL.}  CGAL implements a parallel version of their k-nearest neighbor code using the threading building blocks (TBB)~\cite{TBB}. We use their code directly with no modifications.  We note that their code does not scale well past 16 or so threads. Furthermore, although the code appears to be thread safe, there seems to be contention when there are many threads, thereby slowing them all down. Due to the particularly bad performance beyond 36 threads (which all are on one chip), we only report numbers up to 36 threads. Furthermore, since we observed wildly varying times with higher $k$, we only included times for $k<10$ in our experiments.

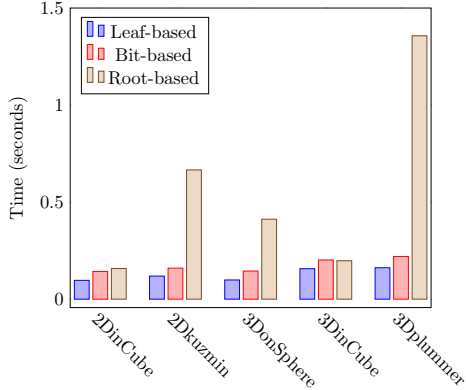
\begin{figure}[t]
\vspace{-.5em}
\begin{center}
  \begin{tikzpicture}[scale = .7]
  \begin{axis}[ylabel  = Time (seconds),
  	legend pos = north west,
    ybar,
    ymax =1.5,
    bar width = 8pt,
    tickwidth         = 0pt,
    symbolic x coords = {2DinCube, 2Dkuzmin, 3DonSphere, 3DinCube, 3Dplummer},
    xticklabel style={
    	anchor=north west,
    	align=left,
    	rotate=-45,
    	inner sep=1pt,
    },
  ]
  \addplot coordinates { 
  	(2DinCube, .097) 
  	(2Dkuzmin, .119) 
  	(3DinCube, .157 )
  	(3DonSphere, .099) 
  	(3Dplummer, .162 )
                  		};
  \addplot coordinates { 
  	(2DinCube, .143 )
  	(2Dkuzmin,  .16 )
  	(3DinCube, .202 )
  	(3DonSphere, .145 )
  	(3Dplummer, .22 )
  };
	\addplot coordinates { 	
		(2DinCube, .158 )
		(2Dkuzmin, .666 )
		(3DinCube, .198 )
		(3DonSphere, .412 )
		(3Dplummer, 1.357 )
	};
  \legend{Leaf-based, Bit-based, Root-based}
  \end{axis}
\end{tikzpicture} 
\end{center}\vspace{-1.5em}
\caption{\small A bar chart showing the performance of our three k-nearest neighbor algorithms on different datasets. All datasets are 10 million points, and times reported are for  $k=1$.}
\label{fig: databarchart}
\vspace{-2em}
\end{figure}
\begin{figure*}[t]
\vspace{-.5em}
	\begin{subfigure}{.66\columnwidth}
%
\begin{tikzpicture}[scale=.65]
		\begin{axis}[
		legend style={font=\footnotesize},
		xmode = log,
		ymode = log,
		xlabel={Size of Dataset},
		ylabel={Time per point (seconds)},
		xmin=0, xmax=100000000,
		ymax=.000004,
		legend pos=north east,
		ymajorgrids=true,
		grid style=dashed,
		]
		
		\addplot[
		color=blue,
		mark = o,
		]
		coordinates {
			(100, 4.5598/10000000)
			(1000, 6.3798/10000000)
			(10000, 2.7526/10000000)
			(100000,.3764 /10000000)
			(1000000,.018/1000000)
			(5000000,.091/5000000)
			(10000000,.178/10000000)
			(100000000,1.681/100000000)
		};
		\addlegendentry{Leaf-based}

		\addplot[
		color=black,
		mark = star,
		]
		coordinates {
		(100, 6.4088/10000000)
		(1000, 4.9645/10000000)
		(10000, 2.7968/10000000)
		(100000,.3778 /10000000)
		(1000000,.021/1000000)
		(5000000,.116/5000000)
		(10000000,.232/10000000)
		(100000000,2.129/100000000)
		};
		\addlegendentry{Root-based}
		
				\addplot[
		color=red,
		mark = triangle,
		]
		coordinates {
			(100, 2.21/1000000)
			(1000, .5700/1000000)
			(10000, .4636/1000000)
			(100000,1.53 /10000000)
			(1000000,.133/1000000)
			(5000000,.622/5000000)
			(10000000,1.251/10000000)
			(100000000,15.035/100000000)
		};
		\addlegendentry{Chan}
		
		\addplot[
		color=green,
		mark = square,
		]
		coordinates {
			(100, 5.7798/10000000)
			(1000, 2.3845/10000000)
			(10000, 1.5866/10000000)
			(100000,.8624 /10000000)
			(1000000,.087/1000000)
			(5000000,.355/5000000)
			(10000000,.729/10000000)
			(100000000,7.889/100000000)
		};
		\addlegendentry{ParlayKNNG}
		
			\addplot[
		color=violet,
		mark = diamond,
		]
		coordinates {
			(100, 8.8377/10000000)
			(1000, 2.6016/10000000)
			(10000, 1.4634/10000000)
			(100000, 1.0277/10000000)
			(1000000, 1.1406/10000000)
			(10000000, 2.1585/10000000)
			(100000000, 24.4608/100000000)
			
		};
		\addlegendentry{CGAL (36 threads)}
		
	\end{axis}
\end{tikzpicture}

		\caption{Time required to calculate nearest neighbors as the size of the dataset increases. Calculated by dividing the total time by the number of points queried.}
	\end{subfigure}~~
	\begin{subfigure}{.66\columnwidth}
%
	
\begin{tikzpicture}[scale=.65]
	\begin{axis}[
		legend style={font=\footnotesize},
		xmode = log,
		ymode = log,
		xlabel={Size of Dataset},
		ylabel={Time per point (seconds)},
		xmin=0, xmax=100000000*1.1,
		ymin=0, ymax=.000004,
		legend pos=north east,
		ymajorgrids=true,
		grid style=dashed,
		]
		
		\addplot[
		color=blue,
		mark = o,
		]
		coordinates {
			(100, 3.4964/10000000)
			(1000,3.7654 /10000000)
			(10000, 2.5280/10000000)
			(100000,.2090 /10000000)
			(1000000,.011/1000000)
			(5000000,.059/5000000)
			(10000000,.118/10000000)
			(100000000,1.08/100000000)
		};
		\addlegendentry{Leaf-based}
		
			\addplot[
		color=black,
		mark = star,
		]
		coordinates {
		(100, 5.3704/10000000)
		(1000, 3.8644/10000000)
		(10000,2.37 /10000000)
		(100000,.2912 /10000000)
		(1000000,.015/1000000)
		(5000000,.078/5000000)
		(10000000,.159/10000000)
		(100000000,1.494/100000000)
		};
		\addlegendentry{Root-based}

		\addplot[
		color=red,
		mark = triangle,
		]
		coordinates {
			(100, 1.6157/1000000)
			(1000, .5279/1000000)
			(10000, .3042/1000000)
			(100000, .8895/10000000)
			(1000000,.052/1000000)
			(5000000,.223/5000000)
			(10000000,.452/10000000)
			(100000000,4.606/100000000)
		};
		\addlegendentry{Chan}
		
			\addplot[
		color=green,
		mark = square,
		]
		coordinates {
			(100, 4.4634/10000000)
			(1000, 1.7875/10000000)
			(10000,1.1190 /10000000)
			(100000,.4342 /10000000)
			(1000000,.029/1000000)
			(5000000,.125/5000000)
			(10000000,.232/10000000)
			(100000000,2.358/100000000)
		};
		\addlegendentry{ParlayKNNG}
		
		\addplot[
		color = violet,
		mark = diamond,
		]
		coordinates{
			(100, 9.1624/10000000)
			(1000,3.2082/10000000)
			(10000,1.6600/10000000)
			(100000,1.2231/10000000)
			(1000000,1.3925/10000000)
			(10000000,2.4450/10000000)
			(100000000,31.1778/100000000)
		};
		\addlegendentry{CGAL (36 threads)}
		
	\end{axis}
\end{tikzpicture}

		\caption{The same as (a) but with a 2D dataset drawn randomly from a square instead of a 3D dataset.}
	\end{subfigure}~~
	\begin{subfigure}{.66\columnwidth}
%
\begin{tikzpicture}[scale=.65]
	\begin{axis}[
		legend style={font=\footnotesize},
		xlabel={Number of Threads},
		ylabel={Work = Threads $\times$ Time},
		xmin=0, xmax=150,
		ymin=0, ymax=230,
		legend pos=north west,
		ymajorgrids=true,
		grid style=dashed,
		]
		
		\addplot[
		color=blue,
		mark = o,
		]
		coordinates {
			(1,9.228)
			(2,4.686*2)
			(8,1.177*8)
			(20,.64*20)
			(40,.305*40)
			(72,.194*72)
			(144,.157*144)
		};
		\addlegendentry{Leaf-based}
		
		\addplot[
		color=black,
		mark = star,
		]
		coordinates {
			(1,14.668)
			(2,7.324*2)
			(8,2.96*8)
			(20,.926*20)
			(40,.445*40)
			(72,.29*72)
			(144,.198*144)
		};
		\addlegendentry{Root-based}
		
		\addplot[
		color=red,
		mark = triangle,
		]
		coordinates {
			(1,97.2300)
			(2,50.5493*2)
			(8,14.0612*8)
			(20,5.5863*20)
			(40,2.9030*40)
			(72,1.8908*72)
			(144,1.2769*144)
		};
		\addlegendentry{Chan}

			\addplot[
		color=green,
		mark = square,
		]
		coordinates {
			(1,48.8569)
			(2,24.7384*2)
			(8,7.1832*8)
			(20,3.1423*20)
			(40,1.7805*40)
			(72,1.0974*72)
			(144,.8080*144)
		};
		\addlegendentry{ParlayKNNG}
		
		\addplot[
		color=violet,
		mark = diamond,
		]
		coordinates {
	(1,10.6273)
	(2,5.7972*2)
	(8,2.8556*8)
	(18,2.1349*18)
	(36,2.1758*36)
	
	};
	\addlegendentry{CGAL}

	\end{axis}
\end{tikzpicture}

		\caption{Total work (threads $\times$ time) required to build a tree of 10 million points, then build the nearest neighbor graph of the point set. Shown as the number of threads vary.}
	\end{subfigure}
	\par\bigskip
	\begin{subfigure}{.66\columnwidth}
%
\begin{tikzpicture}[scale=.65]
	\begin{axis}[
		legend style={font=\footnotesize},
		xlabel={Number of Threads},
		ylabel={Work = Threads $\times$ Time},
		xmin=0, xmax=65,
		ymin=0, ymax=320,
		ytick = {25, 50, 100,200,300},
		legend pos=north west,
		ymajorgrids=true,
		grid style=dashed,
		]
		
		\addplot[
		color=blue,
		mark = o,
		]
		coordinates {
			(1,18.6562)
			(2,9.7191*2)
			(8,2.5124*8)
			(16,1.2879*16)
			(32,.6759*32)
			(64,.4497*64)
		};
		\addlegendentry{Leaf-based}

		\addplot[
		color=black,
		mark = star,
		]
		coordinates {
			(1,32.6442)
			(2,16.6585*2)
			(8,4.3058*8)
			(16,2.2200*16)
			(32,1.2791*32)
			(64,.8769*64)
		};
		\addlegendentry{Root-based}
		
		\addplot[
		color=red,
		mark = triangle,
		]
		coordinates {
			(1,173.5695)
			(2,91.0046*2)
			(8,22.9263*8)
			(16,11.5896*16)
			(32,6.1379*32)
			(64,3.7365*64)
		};
		\addlegendentry{Chan}

			\addplot[
		color=green,
		mark = square,
		]
		coordinates {
			(1,111.6814)
			(2,55.3072*2)
			(8,13.5542*8)
			(16,6.8385*16)
			(32,3.6376*32)
			(64,2.1592*64)
		};
		\addlegendentry{ParlayKNNG}
		
	\end{axis}
\end{tikzpicture}

		\caption{The same measurements as (c), but on the 32-core AMD machine.}
	\end{subfigure} ~~
	\begin{subfigure}{.66\columnwidth}
%
		\begin{tikzpicture}[scale=.65]
			\begin{axis}[
				legend style={font=\footnotesize},
				ymode = log,
				xlabel={$k$ (number of nearest neighbors)},
				ylabel={Time per neighbor (seconds)},
				xmin=0, xmax=100,
				legend pos=north east,
				ymajorgrids=true,
				grid style=dashed,
				]
				
				\addplot[
				color=blue,
				mark = o,
				]
				coordinates {
					(1,.203/10000000)
					(2, .214/20000000)
					(5,.315/50000000)
					(10,.453/100000000)
					(15, .586/150000000)
					(20, .678/200000000)
					(25, .828/250000000)
					(30,.904/300000000)
					(40,1.215/400000000)
					(50,1.562/500000000)
					(60,2.056/600000000)
					(70,2.244/700000000)
					(80,2.422/800000000)
					(90,2.646/900000000)
					(100,2.781/1000000000)
					
				};
				\addlegendentry{Leaf-based}

				\addplot[
				color=black,
				mark = star,
				]
				coordinates {
				(1,.254/10000000)
				(2, .294/20000000)
				(5,.35/50000000)
				(10,.469/100000000)
				(15, .523/150000000)
				(20, .626/200000000)
				(25, .78/250000000)
				(30,.88/300000000)
				(40,1.175/400000000)
				(50,1.655/500000000)
				(60,1.928/600000000)
				(70,2.327/700000000)
				(80,2.453/800000000)
				(90,2.619/900000000)
				(100,2.884/1000000000)
				
				};
				\addlegendentry{Root-based}
				
				\addplot[
				color=green,
				mark = square,
				]
				coordinates {
					(1,.785/10000000)
					(2,.957/20000000)
					(5,1.318/50000000)
					(10,1.711/100000000)
					(15,2.031/150000000)
					(20,2.404/200000000)
					(25,2.715/250000000)
					(30,3.035/300000000)
					(40,3.731/400000000)
					(50,4.359/500000000)
					(60,5.155/600000000)
					(70,5.903/700000000)
					(80,6.809/800000000)
					(90,7.694/900000000)
					(100,8.682/1000000000)
				};
				\addlegendentry{ParlayKNNG}
				
						\addplot[
				color=violet,
				mark = diamond,
				]
				coordinates {
			(1, 2.1758/10000000)
			(2, 2.5704/20000000)
			(5, 3.4485/50000000)
			(10, 4.0751/100000000)	
				};
				\addlegendentry{CGAL (36 threads)}

			\end{axis}
		\end{tikzpicture}

		\caption{Time required to calculate a neighbor as the number of neighbors $k$ increases. Calculated by dividing the total time by $k$ times the number of queries.}
	\end{subfigure}~~
	\begin{subfigure}{.66\columnwidth}
%

\begin{tikzpicture}[scale=.65]
  \begin{axis}[ylabel = Time(seconds),
  	legend pos = north west,
    ybar,
    ymode = log,
    log origin=infty,
    ymax =20,
    tickwidth         = 0pt,
    symbolic x coords = {zd-tree, parlayKNN, Chan05, CGAL},
    xtick = data,
 xticklabel style={
	anchor=north west,
	align=left,
	rotate=-45,
	inner sep=1pt,
},
  ]
  \addplot coordinates { 
	(zd-tree, .0714)
	(parlayKNN, .0635)
	(Chan05, .1901)
	(CGAL, 3.5338)
};
  \addplot coordinates { 
  	(zd-tree, .1007)
  	(parlayKNN, .0667)
  	(Chan05, .2280)
  	(CGAL, 19.7503)
  };

  \legend{3DinCube, 3Dplummer}
  \end{axis}
\end{tikzpicture}

		\caption{The bars represent the total time each algorithm takes to build the data structure, for points drawn randomly from a 3D cube, and points drawn from a Plummer distribution.}
	\end{subfigure}

  \vspace{-.5em}
\caption{\small Statistics related to non-dynamic queries. Unless otherwise stated, the size of the dataset is 10 million, the number of nearest neighbors $k=1$, experiments were performed on 144 threads on a 72-core Dell R930, and data points are drawn randomly from a 3D cube. }
\label{fig: nondynqueries}
\vspace{-1em}
\end{figure*}
\begin{figure}[t]
\vspace{-.5em}
%
\begin{tikzpicture}[scale=.8]
	\begin{axis}[
		xlabel={Number of points added dynamically},
		ylabel={Time per point (seconds)},
		xmode = log,
		ymode = log,
		legend pos=north east,
		ymajorgrids=true,
		grid style=dashed,
		]
		
		\addplot[
		color=blue,
		]
		coordinates {
		(1, 167.6289/5000000)
		(10, 47.0994/5000000)
		(100, 26.5538/5000000)
		(1000, 3.2726/5000000)
		(10000, 1.3916/5000000)
		(100000, .2173/5000000)
		(1000000, .0924/5000000)
		(50000000, .0654/5000000)
		};\addlegendentry{Insertions}
		
			\addplot[
		color=red,
		]
		coordinates {
			(1, 297.8784/5000000)
			(10, 60.7192/5000000)
			(100, 25.6942/5000000)
			(1000, 3.5176/5000000)
			(10000, 1.7022/5000000)
			(100000, .2661/5000000)
			(1000000,.0968 /5000000)
			(50000000, .0682/5000000)
			
		};\addlegendentry{Deletions}

	\end{axis}
\end{tikzpicture}

\caption{\small Time required per point as the number of points in the batch increases. Updates performed on a dataset of 10 million random points inside a 2D cube.}
\label{fig: dynamicupdates}
\vspace{-2em}
\end{figure}
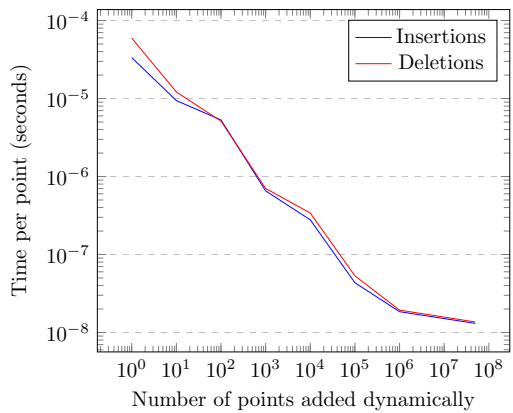
\section{Experiments}\label{sec: experiments}
In this section, we provide experimental results which show that 1) our algorithms perform well under many types of scaling and across different architectures, and 2) our algorithms outperform every implementation we test against.

\subsection{Experimental Setup.}\label{subsec: experimentsetup}

\paragraph{Machines.} We ran most all of our experiments on a 72-core Dell R930 with 4x Intel(\textregistered) Xeon(\textregistered) E7-8867 v4 (18 cores, 2.4GHz and 45MB L3 cache), and 1Tbyte memory. With hyperthreading the total number of threads is 144.
To check whether the results were robust across machines, we also ran one set of experiments on a 4-socket AMD machine with 32 physical cores in total, each running at 2.4 GHz, with 2-way hyperthreading, a 6MB L3 cache per socket, and 200 GB of main memory.

\paragraph{Test data.} After testing our algorithms on some real-world image data, we discovered, similarly to Connor and Kumar~\cite{connor2010knn}, that uniformly random points perform very similarly to real-world datasets.  To facilitate testing at various sizes we therefore use a few distributions of random points in 2 and 3 dimensions, on sizes up to 100 million points. The point sets we used are listed in Figure~\ref{fig: databarchart}. The 2D and 3DinCube datasets are points picked uniformly at random in a square and cube, respectively. The points in the 3DonSphere dataset are selected on the 2D surface of a sphere in 3 space. This is meant to represent various graphics applications where the point sets are on a 2D surface embedded in 3D. The 3Dplummer distribution uses the Plummer model~\cite{AarsethHW74}, which is based on the study of galaxies, and highly dense in at the center, becoming very sparse on the outside.  The 2Dkuzmin distribution is a similarly skewed distribution in two dimensions.

As can be seen from the various statistics, the Plummer and Kuzmin distributions are significantly more skewed than the others. Indeed. these distributions do not have bounded expansion constant. The performance is slower due to the fact that a few small points are extremely far away from most of the points, which are in a dense cluster at the center. This causes the tree to be unbalanced and the searches for the nearest neighbors of these far points to be expensive.  However the overall time is hardly affected by the skewed distribution for our leaf and bit-based implementations (Figure~\ref{fig: databarchart}) showing the algorithms are robust under quite skewed distributions.  As expected, the depth of the trees for the uniform distributions in a cube are just the logarithm of the number of points.

\paragraph{Algorithms tested.} We ran three classes of experiments: (1) generating a $k$-nearest neighbor graph on a set of $n$ points, (2) building a $k$-nearest neighbor query structure on $n$ points followed by dynamic queries on a different set of $n$ points, and (3) batch insertions for a total of $n$ points (after the insertion). The results from (2) can be found in the full version of our paper; they are not significantly different from building the k-nearest neighbor graph.

Altogether we tested 9 variants of the algorithms: our parallelized version of Chan's algorithm, the CGAL algorithm, four variants of STANN (KNN, KNNG, parlayKNN and parlayKNNG), and three variants of our algorithm (leaf-based, root-based, and bit-based). The parlayKNN and parlayKNNG are our modified versions of Connor and Kumar's algorithms. Since our modified versions are always significantly faster, we only report numbers for our versions.

For all implementations, we sort by Morton order before querying. This is to ensure that all algorithms are getting the same benefit of locality in the tree when querying. The experiments on batch insertion only use our algorithm since the others do not support dynamic updates.

\subsection{Leaf vs. Root Based.}\label{subsec: leafvroot} 
In Figure~\ref{fig: databarchart}, we show the performance of our three search algorithms for finding the k-nearest neighbor graph on varying datasets and $k=1$. Figure~\ref{fig: databarchart} shows the same result using a different measurement: the average and maximum number of nodes visited during a query. One takeaway from the figures is that even though the leaf-based method takes $O(n)$ work, and the bit-based method takes $O(n \log n)$ work, the prior is only slightly faster.  This is because the constant in the $O(k \log k)$ term for a search from the leaf is much larger than the constant in the $O(\log n)$ search from the root.  Another takeaway is that for the Kuzmin and Plummer distributions, when starting from the leaf using the root-based algorithm (Algorithm~\ref{algo: naivesearch}) makes an enormous difference.

\subsection{$k$-Nearest Neighbor Graphs.}\label{subsec: knng} 

The results of our experiments for generating the $k$-nearest neighbor graph can be found in Figure~\ref{fig: nondynqueries}.

\paragraph{Varying dataset size.} (Figure~\ref{fig: nondynqueries}(a) and (b).) We measured the total time per point (that is, to build the tree and perform the query) by dividing the total time to build and search by the number of points. As discussed in Section~\ref{sec: other_impls}, we took measures to limit the number of cache misses where possible.

\paragraph{Work efficiency.} (Figure~\ref{fig: nondynqueries}(c) and (d).) Experiments showed that our algorithms performed significantly less work than our competitors as the number of threads increased to 144. To show that we maintain work efficiency on different architectures, we also ran the same experiments on a 32-core AMD machine.

\paragraph{Varying $k$.} (Figure~\ref{fig: nondynqueries}(e).) Our results on varying numbers of neighbors show that our algorithms remain fast and scalable. 

\paragraph{Tree building.} (Figure~\ref{fig: nondynqueries}(f).) To illustrate that the tree-building step itself is efficient (except in the case of CGAL as explained in Section~\ref{subsec: other_impls}), we show time required to build the data structure for the 3DinCube and 3Dplummer distributions.

\subsection{Dynamic Updates.}
\label{subsec: dynupdates} 
We test the efficiency of our batch-dynamic updates by measuring the time required per update as the number of updates in the batch increases. Figure~\ref{fig: dynamicupdates} shows the time taken per point as the size of the batch increases, with both insertions and deletions shown. The figure shows a drastic change in time, spanning almost four orders of magnitude from $10^{-4}$ seconds for a single update to $10^{-8}$ seconds per update for a batch of 5 million. The first period of decrease as the batch size increases to $10^4$ or $10^5$ can be explained by parallelism---this is the point at which the parallel sort and the parallel recursion down the tree begin to save significant time. The fact that time continues to decrease even after the size grows large enough to see the full effects of parallelism can be attributed to the work efficiency of the batch-dynamic updates, as shown in Theorem~\ref{thm: batchupdate}.

\subsection{Code availability.}
Our implementation is part of the publicly available Problem-Based Benchmark Suite~\cite{shun2012pbbs}.
\section{Conclusion}\label{sec: conclusion}
In this work, we presented the zd-tree, a data structure for k-nearest neighbors that combines the ideas of kd-trees and Morton ordering and supports batch-dynamic updates. We showed that the zd-tree is both theoretically efficient and fast in practice, performing well even on data sets which do not have bounded ratio or bounded expansion. One future experimental direction is to experiment with datasets with higher dimensions, or with using our algorithms as a sub-step in calculating nearest neighbors in high dimensions. Another direction worth exploring is to use the zd-tree or a similar data structure for other problems in low-dimensional geometry, such as closest pair or n-body interactions.

\section*{Acknowledgments}
We thank the anonymous referees for their comments and suggestions. This research was supported by NSF grants CCF-1901381, CCF-1910030, and CCF-1919223, and the NSF GRFP.

\small
\bibliographystyle{abbrv}
\bibliography{main.bbl}

\newpage

\appendix

\section{Other Nearest Neighbor Implementations}\label{apdx: otherimpls}

For the purpose of comparison we use three existing implementations of nearest neighbor search: CGAL~\cite{alliez2016cgal}, STANN~\cite{connor2010knn}, and Chan~\cite{chan2006minimalist}. Here we describe some performance issues with their code, and some modifications we made to improve the performance of their code.

Furthermore, we give a brief explanation for why we do not benchmark Delaunay triangulation based methods, since those are sometimes used for computing nearest neighbors in low dimensions. Firstly, the existing implementations for finding Delaunay triangulations were very slow: on 10 million points, the ParlayLib built-in function takes a few seconds, compared to less than .05 seconds to build the kd-tree. Secondly, the literature we found on Delaunay triangulations focused on the 2D case~\cite{birn2010simple}, while our experiments focused on the 3D case. 

\paragraph{Chan.}  Chan's code was fully sequential so we needed to parallelize it.  Conceptually this is relatively straightforward since the algorithm just requires using a parallel sort instead of a sequential one, and then running the queries in parallel.  Indeed the first step of using a parallel sort was easy and we just replaced the C++ STL sort with the ParlayLib sort.  The second step was required some work since the code was not thread safe.  However, once modified and using the parallel loop from ParlayLib the code achieves very good speedup---about 75-fold speedup on 72 cores with 144 threads (see Section~\ref{sec: experiments} for the full details).

Chan's algorithm only describes how to search from the root, and correspondingly his code only searches from the root.  There seems to be no inherent reason that it would not be possible to start at the leaves when generating a nearest neighbor graph, but we did not implement such a variant.  We note that even the root-based implementation of our code is significantly faster than Chan's. Chan's code uses arrays to store the points and therefore does not support dynamic updates.  In his paper he mentions that his algorithm could support dynamic updates by storing the points in a balanced binary search tree in Morton order. This would require completely rewriting their code. Experiments with Chan's code only deal with the $k=1$ case since his algorithms did not provide support for higher $k$. 

\paragraph{STANN.}  STANN includes both a k-nearest neighbor graph (KNNG) function and a $k$-nearest neighbor (KNN) function.  The first finds the $k$ nearest neighbors among a set of points, and the second supports a function to build a tree and a separate function to query a point for its $k$-nearest neighbors.  They supply a parallel version of KNNG, that was parallelized with OpenMP, and only a sequential version of KNN.  On our initial tests we were able to get performance on the parallel KNNG that closely match what they report.  However, their algorithm did not scale well beyond 16 threads (their numbers agree with this).  The main issues is that the algorithm left some components sequential, including the Morton order sort, and initializing various standard template library (STL) vectors.  For a larger number of cores this became the bottleneck.  We therefore updated their code to use ParlayLib using a parallel sort and replacing uses of STL vectors with ParlayLib sequences, which are initialized in parallel.  We also made a couple other optimizations, including changing the size of the base case of the recursive query from 4 to 10, and using vectors instead of a priority queues to store the nearest neighbors for a point when $k$ is small.  These changes made a significant improvement in performance, especially at a larger number of cores, as indicated in 
Figure~\ref{fig: STANNvParlay}.

The STANN KNN code was fully sequential.  We therefore parallelized this as well, which required much the same changes as we made to the parallel code (using a parallel sort, making all loops parallel, and using parlay sequences instead of STL vectors).  We also run the queries in parallel, which required some minor changes to make their code thread safe.  As with Chan's algorithm, STANN stores the points in an array (STL vector) and therefore does not support dynamic updates.

\paragraph{CGAL.}  CGAL implements a parallel version of their $k$-nearest neighbor code using the threading building blocks (TBB)~\cite{TBB}.  We use their code directly with no modifications.  We note that as with the original version of STANN KNNG, their code does not scale well past 16 or so threads.  We looked into this and there are several reasons.  
Perhaps most fundamentally, although in their recursive routine for building the kd-tree they invoke the two recursive calls in parallel, they do the splitting within each node completely sequentially.  At the root of the tree this means they do linear work completely sequentially.  Indeed from a theoretical point of view their algorithm does a total of $O(n \log n)$ work (assuming uniform input) and has $O(n)$ span, meaning it only has $O(\log n)$ parallelism.  Fixing this problem would require a major rewrite of their code.   

A second issue is that they allocate their tree nodes by pushing onto the back of a TBB concurrent vector.   Although this is thread safe, it requires a lock and becomes a bottleneck on a large number of threads.  
A similar issue appears to be true in the query.   In particular although the code appears to be thread safe giving correct answers when run in parallel, there seems to be contention when there are many threads slowing them all down.   This is often caused by some form of memory allocation, as with the build tree, but in this case we were not able to track down the source of the problem.    Due to the particularly bad performance beyond 36 threads (which all are on one chip), we only report numbers up to 36 threads. Furthermore, since we observed wildly varying times with higher $k$, we only included times for $k<10$ in our experiments. 

\section{Dynamic Queries}\label{apdx: dynqueries}

In this appendix, we provide data and commentary on our experimental results for non-dynamic queries. The results are presented graphically in Figure~\ref{fig: nondynqueries}.

\paragraph{Varying $k$, work efficiency.} (Figure~\ref{fig: dynamicupdates}(a) and (b).) All the algorithms we tested had similar performance for work efficiency and scalability of the number of nearest neighbors as they did in the k-nearest neighbor graph building case.

\paragraph{Scaling size of dataset.} (Figure~\ref{fig: dynamicupdates}(c).) Our algorithms are particularly robust compared to others when scaling the size of the dataset. One thing to notice in the relevant panel (c) of Figure~\ref{fig: dynqueries} is that with the exception of CGAL, all algorithms experience a local minima when the size of the dataset reaches $10^5$. While Chan, ParlayKNN, and CGAL's performances begin to slowly increase after this point, our algorithms' do not. The dataset of size $10^5$ is approximately where we expect cache misses to start affecting the performance; as explained in Section~\ref{sec: other_impls}, pre-sorting the data for dynamic queries helps alleviate this problem. 

\paragraph{Other datasets.} Since the bit-based and root-based algorithms performed very similarly on the random distribution for dynamic queries, we refer to Figure~\ref{fig: nondynqueries}(f) for evidence that the bit-based algorithm outperforms the root-based for some distributions.

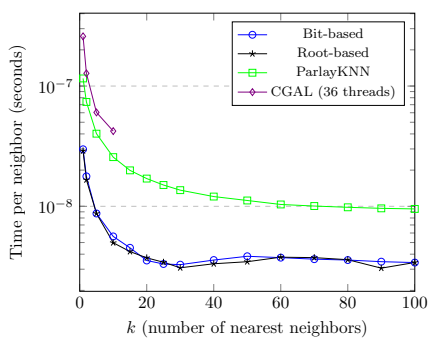
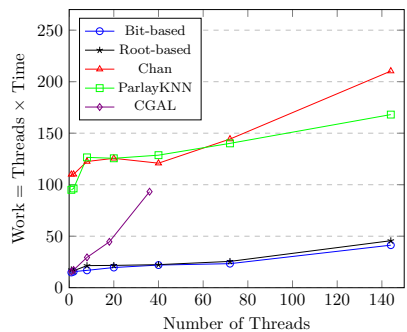
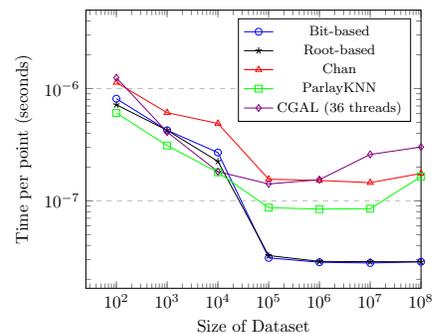
\begin{figure*}[t]
\vspace{-.5em}
	\begin{subfigure}{.66\columnwidth}
%
		\begin{tikzpicture}[scale = .65]
			\begin{axis}[
				legend style={font=\footnotesize},
				ymode = log,
				xlabel={$k$ (number of nearest neighbors)},
				ylabel={Time per neighbor (seconds)},
				xmin=0, xmax=100,
				legend pos=north east,
				ymajorgrids=true,
				grid style=dashed,
				]
				
				\addplot[
				color=blue,
				mark = o,
				]
				coordinates {

					(1, .298/10000000)
					(2, .355/20000000)
					(5, .436/50000000)
					(10, .562/100000000)
					(15, .678/150000000)
					(20, .709/200000000)
					(25, .829/250000000)
					(30, .982/300000000)
					(40, 1.43/400000000)
					(50,1.922 /500000000)
					(60, 2.247/600000000)
					(70, 2.546/700000000)
					(80, 2.861/800000000)
					(90, 3.12/900000000)
					(100, 3.407/1000000000)
					
				};
				\addlegendentry{Bit-based}

				\addplot[
				color=black,
				mark = star,
				]
				coordinates {
					(1, .289/10000000)
					(2, .332/20000000)
					(5, .436/50000000)
					(10, .498/100000000)
					(15, .631/150000000)
					(20, .746/200000000)
					(25, .857/250000000)
					(30, .925/300000000)
					(40, 1.332/400000000)
					(50, 1.732/500000000)
					(60, 2.266/600000000)
					(70, 2.622/700000000)
					(80, 2.88/800000000)
					(90, 2.756/900000000)
					(100, 3.412/1000000000)
				};
				\addlegendentry{Root-based}
				
				\addplot[
				color=green,
				mark = square,
				]
				coordinates {
					(1, 1.157/10000000)
					(2, 1.475/20000000)
					(5, 2.009/50000000)
					(10, 2.565/100000000)
					(15, 2.987/150000000)
					(20, 3.402/200000000)
					(25,3.759/250000000)
					(30, 4.089/300000000)
					(40, 4.822/400000000)
					(50, 5.592/500000000)
					(60,6.217/600000000)
					(70, 7.039/700000000)
					(80,7.853/800000000)
					(90, 8.659/900000000)
					(100, 9.516/1000000000)
					
				};
				\addlegendentry{ParlayKNN}
				
						\addplot[
				color=violet,
				mark = diamond,
				]
				coordinates {
				(1, 2.5896/10000000)
				(2, 2.5498/20000000)
				(5, 3.0186/50000000)
				(10, 4.2227/100000000)

				};
				\addlegendentry{CGAL (36 threads)}

			\end{axis}
		\end{tikzpicture}

			\caption{Time required to calculate a neighbor as $k$ increases. Calculated by dividing the total time by $k$ times the number of queries.}
	\end{subfigure}~~
	\begin{subfigure}{.66\columnwidth}
%
\begin{tikzpicture}[scale = .65]
	\begin{axis}[
		xlabel={Number of Threads},
		ylabel={Work = Threads $\times$ Time},
		legend style={font=\footnotesize},
		xmin=0, xmax=150,
		ymin=0, ymax=270,
		legend pos=north west,
		ymajorgrids=true,
		grid style=dashed,
		]
		
		\addplot[
		color=blue,
		mark = o,
		]
		coordinates {
			(1,14.8089*1)
			(2,7.8032*2)
			(8, 2.1096*8)
			(20, .9804*20)
			(40, .5501*40)
			(72, .3245*72)
			(144, .2869*144)
		};
		\addlegendentry{Bit-based}
		
					\addplot[
		color=black,
		mark = star,
		]
		coordinates {
			
			(1,17.3764*1)
			(2 ,8.2326*2)
			(8, 2.6835*8)
			(20, 1.0835*20)
			(40, .5599*40)
			(72, .3550*72)
			(144,  .3169*144)

		};
		\addlegendentry{Root-based}
		
		\addplot[
		color=red,
		mark = triangle,
		]
		coordinates {
			
			(1,110.0163*1)
			(2,55.0638*2)
			(8, 15.3169*8)
			(20, 6.2804*20)
			(40, 3.0225*40)
			(72, 2.0043*72)
			(144,1.4604*144)

		};
		\addlegendentry{Chan}

			\addplot[
		color=green,
		mark = square,
		]
		coordinates {
			(1,95.0146)
			(2, 48.1828*2)
			(8, 15.8163*8)
			(20, 6.2823*20)
			(40, 3.2169*40)
			(72, 1.9462*72)
			(144, 1.1671*144)

		};
		\addlegendentry{ParlayKNN}
		
		\addplot[
		color=violet,
		mark = diamond,
		]
		coordinates {
		(1, 15.6846)
		(2, 8.6082*2)
		(8, 3.6818*8)
		(18,2.4750*18)
		(36,2.5896*36)

	};
	\addlegendentry{CGAL}

	\end{axis}
\end{tikzpicture}

			\caption{Total work (threads $\times$ time) required to build a tree of 10 million points, then dynamically query the same number of points, on varying numbers of threads.}
	\end{subfigure}~~
	\begin{subfigure}{.66\columnwidth}
%
\begin{tikzpicture}[scale = .65]
		\begin{axis}[
		legend style={font=\footnotesize},
		xmode = log,
		ymode = log,
		xlabel={Size of Dataset},
		ylabel={Time per point (seconds)},
		xmin=0, xmax=100000000,
		ymax=.000005,
		legend pos=north east,
		ymajorgrids=true,
		grid style=dashed,
		]
		
		\addplot[
		color=blue,
		mark = o
		]
		coordinates {
			(100, 8.1246 /10000000)
			(1000, 4.253/10000000)
			(10000, 2.6935/10000000)
			(100000, .3103/10000000) 
			(1000000, .2838/10000000)
			(10000000, .2795/10000000)
			(100000000, 2.8652/100000000)
			
		};
		\addlegendentry{Bit-based}

		\addplot[
		color=black,
		mark = star,
		]
		coordinates {
			(100, 7.1787 /10000000)
			(1000, 4.2523/10000000)
			(10000 ,2.2482 /10000000)
			(100000, .3256/10000000) 
			(1000000, .2887 /10000000)
			(10000000, .2871/10000000)
			(100000000, 2.8632/100000000)
			
		};
		\addlegendentry{Root-based}
		
		\addplot[
		color=red,
		mark = triangle,
		]
		coordinates {
			(100, 11.3375/10000000)
			(1000, 6.1020/10000000)
			(10000, 4.8886/10000000)
			(100000, 1.5580/10000000) 
			(1000000, 1.5204/10000000)
			(10000000, 1.4528/10000000)
			(100000000, 17.5379/100000000)

		};
		\addlegendentry{Chan}

		\addplot[
		color=green,
		mark = square,
		]
		coordinates {
		(100, 6.0628/10000000)
		(1000, 3.1095/10000000)
		(10000,1.7957 /10000000)
		(100000, .8700/10000000) 
		(1000000, .8439/10000000)
		(10000000, .8518/10000000)
		(100000000, 16.4453/100000000)
		
		};
		\addlegendentry{ParlayKNN}
		
			\addplot[
		color=violet,
		mark = diamond,
		]
		coordinates {
			(100, 12.4963/10000000)
			(1000,4.1030/10000000)
			(10000,1.8211/10000000)
			(100000,1.4144/10000000)
			(1000000,1.5456/10000000)
			(10000000,2.5896/10000000)
			(100000000, 30.1612/100000000)
			
		};
		\addlegendentry{CGAL (36 threads)}
		
	\end{axis}
\end{tikzpicture}

			\caption{Time required to calculate nearest neighbors as the size of the dataset increases. Calculated by dividing the total time by the number of points queried.}

	\end{subfigure}
\vspace{-.5em}
\caption{\small Statistics related to dynamic queries. Unless otherwise stated, the size of the dataset is 10 million, 10 million dynamic queries were performed, the number of nearest neighbors $k=1$, experiments were performed on 144 threads on a 72-core Dell R930, and data points are drawn randomly from a 3D cube.}
\label{fig: dynqueries}
\vspace{-1em}
\end{figure*}

\section{Full Proof of Theorem~\ref{thm: batchupdate}}\label{apdx: fulltheory}

Here we extend the results in the main body where $X$ is a general convex space rather than a bounding cube. This is only a concern for Theorem~\ref{thm: batchupdate}; Theorems~\ref{thm: treebuild} and~\ref{thm: querytime} did not actually use the assumption on $X$.

Theorem~\ref{thm: batchupdate} works by giving an upper bound on the number of points in an update that can be inserted into an unbalanced path of length $\ell$. To generalize it to arbitrary convex spaces, we need to show that the upper bound a) still holds when the bounding box $B$ of $X$ is not a hypercube and b) when $X$ is not a hypercube or hyperrectangle. The following lemma concerns the former property.

\begin{lemma}\label{lem: squaremax}
The number of points in $X$ which can be part of an unbalanced path of length $\ell$ in $T$ is maximized when the bounding box $B$ of $X$ is a hypercube.
\end{lemma}

\begin{proof}
This follows from noting that the proof of Theorem~\ref{thm: treebuild} treated a split of $X$ as a split of an oct- or quad-tree---that is, each split is a $d$-dimensional split of a bounding cube into $2^d$ boxes. Thus, if the longest length in the bounding hyperrectangle $B$ of $X$ is normalized to $n$ and another side of $B$ has length $r<<n$, there will only be $\log r$ nonempty cuts of $X$ perpendicular to that side. Thus the length of an unbalanced path in $T$ will always be dominated by the length of the longest edge, and $B$ being a bounding rectangle can only result in fewer points in an update traversing an unbalanced chain than if $B$ were a hypercube.
\end{proof}

The second piece needed to apply Theorem~\ref{thm: batchupdate} to general convex spaces $X$ is to verify that the number of points near the boundary of $X$ is still upper bounded when $X$ is not a hyperrectangle.

\begin{lemma}\label{lem: volmax}
Let $X$ be convex with bounded expansion. Consider a subset of $X$ with volume $V$ such that volume $V$ is a fraction $f$ of $X$'s volume. Then $V$ contains at most $\left(\frac{\gamma^2}{\gamma^2+1} \right)^{\log 1/f}n$ points.
\end{lemma}

\begin{proof}
$O(\log 1/f)$ divisions of $X$ in half are needed to produce a shape of volume $V$. To maximize the number of points in $V$, we need to minimize the number of points that can be in the ``sparse" half of each cut. The number of points in the sparse half is dictated by how many individual boxes can be packed into the sparse area and then expanded twice to reach the boundary of the dense area. Thus, the number of points in the dense area is maximized when there is only one such box, i.e. when $d=1$. This implies that we can upper bound the number of points in $V$ by its upper bound when $d=1$, i.e. a $\left( \frac{\gamma^2}{\gamma^2+1} \right)^{\log 1/f}$ fraction of the total points. This bound will be sufficient to finish the proof of Theorem~\ref{thm: batchupdate}.
\end{proof} 

Now we are ready to prove the extended version of Theorem~\ref{thm: batchupdate}.

\begin{proof}[Proof of Theorem~\ref{thm: batchupdate}]
Fix a bounding box $B$. Now, for all convex shapes with bounding box $B$ and more sides than a hyperrectangle, the surface-area-to-volume ratio decreases from that of a hyperrectangle. Thus, for a given distance to the boundary of $X$, the volume in $X$ that can be at that distance or closer is only smaller when $X$ has more sides than a hyperrectangle. Since Lemma~\ref{lem: volmax} tells us that the upper bound given in Lemma~\ref{lem: slicedensity} applies no matter how the volume is arranged, the number of points in an update that can travel down an unbalanced path is no more than it would be if $X$ were a hyperrectangle. Thus, we can take the number of points (and their corresponding cost) in an update to a hyperrectangle as an upper bound for this case.

However, triangles, tetrahedra, and their higher-dimensional counterparts have a higher surface-area-to-volume ratio than the hyperrectangle. Since the dimension is constant and the bounding box $B$ is assumed to be the smallest possible, the surface-area-to-volume ratio of the tetrahedron is only a constant factor larger than that of the hypercube, and thus the result still holds.
\end{proof}

\end{document}